\documentclass[a4]{article}

\usepackage{times}
\usepackage{tikz}
\usetikzlibrary{arrows,snakes,backgrounds}
\usepackage{amsmath}
\usepackage{amsthm}
\usepackage{amssymb,amsfonts,amsxtra}
\usepackage{stmaryrd}
\usepackage{empheq}
\usepackage{color}
\usepackage{hyperref}
\usepackage{cleveref}
\usepackage[a4paper]{geometry}

\usepackage{thmtools}
\usepackage{thm-restate}

 \theoremstyle{plain}
 \newtheorem{theorem}{Theorem}[section]

 \theoremstyle{definition}

 \newtheorem{example}[theorem]{Example}

 \theoremstyle{remark}

 \numberwithin{equation}{section}
 
 \theoremstyle{remark}

\def \tuple#1{\langle #1 \rangle}

\newcommand{\bN}{\mathbb{N}}

\newcommand{\cS}{\mathcal{S}}

\newcommand{\cG}{\mathcal{G}}
\newcommand{\cE}{\mathcal{E}}

\newcommand{\cC}{\mathcal{C}}
\newcommand{\cD}{\mathcal{D}}
\newcommand{\cP}{\mathcal{P}}
\newcommand{\bT}{\mathbb{T}}
\newcommand{\bZ}{\mathbb{Z}}
\newcommand{\bP}{\mathbb{P}}
\newcommand{\id}{\mathrm{id}}
\newcommand{\ud}{\triangleq}

\newcommand{\Lra}{\Leftrightarrow}
\newcommand{\Ra}{\Rightarrow}
\newcommand{\La}{\Leftarrow}

\newcommand{\ra}{\rightarrow}

\newcommand{\even}{{\textit{even}}}
\newcommand{\odd}{{\textit{odd}}}

\DeclareMathOperator{\Sign}{Sign}

\DeclareMathOperator{\sq}{sq}

\newcommand{\cd}[1]{\text{\lstinline$#1$}}

\DeclareMathOperator{\prt}{prt}

\DeclareMathOperator{\succe}{succ}
\DeclareMathOperator{\parity}{parity}

\newcommand{\wpd}{{\wp^{\downarrow}}}
\newcommand{\etav}{{\eta^{\vee}}}
\newcommand{\alphasing}{{\alpha^{\scriptscriptstyle\{\!\cdot\!\}}}}
\def\sing#1{{#1^{\scriptscriptstyle\{\!\cdot\!\}}}}
\newcommand{\srhd}{{\scriptscriptstyle \rhd}}

\newcommand{\comment}[1]{}

\title{A Constructive Framework for Galois Connections}
\author{Francesco Ranzato\\
{\normalsize Dipartimento di Matematica, University of Padova, Italy}}
\date{}

\begin{document}
\maketitle

\begin{abstract}
Abstract interpretation-based static analyses rely on abstract domains of program properties, such as 
intervals or congruences for integer variables.  
Galois connections (GCs) between posets provide the most widespread and useful formal tool for mathematically 
specifying abstract domains. 
Recently, Darais and Van Horn [2016] put forward a notion of constructive Galois connection for unordered sets (rather than posets), 
which allows to define abstract domains in a so-called mechanized and calculational proof style and therefore enables the
use of proof assistants like Coq and Agda for automatically extracting verified algorithms of static analysis. 
We show here that constructive GCs are isomorphic, in a precise and comprehensive meaning including sound abstract functions, 
to so-called partitioning GCs~---~an already known class of GCs which allows to cast standard set partitions as an abstract domain.
Darais and Van Horn [2016] also provide a notion of constructive GC for posets, which we prove to be
isomorphic to plain GCs and therefore lose their constructive attribute. 
Drawing on these findings, 
we put forward and advocate the use of purely partitioning GCs, a novel  
class of constructive abstract domains for a mechanized approach to abstract interpretation. We show that this class of abstract domains
allows us to represent a set partition with more flexibility while retaining a constructive approach to Galois connections.
\end{abstract}

\section{Introduction}
Abstract interpretation \cite{CC77,CC79} is probably the most used and successful technique for defining
approximations of program  semantics (or, more in general, of computing systems) to be used for designing provably 
sound static program analyzers. 
Abstract domains play a crucial role in any abstract interpretation, since they encode, both logically for reasoning purposes
and practically for implementations,  
which program properties are computed by a static analysis. Since its beginning~\cite{CC77}, 
one major insight of abstract interpretation is given by the use of Galois connections (GCs) for defining abstract domains. A specification
of an abstract domain $D$ 
through a Galois connection  prescribes that: (1)~both concrete and abstract domains, $C$ and $D$, are
partially ordered, and typically they give rise to complete lattices; (2)~definitions of abstraction $\alpha:C\ra D$
and concretization $\gamma:D\ra C$ maps to relate concrete and abstract values; (3)~$\alpha$ and $\gamma$ 
give rise to an adjunction relation: $\alpha(c)\leq_D d \Lra c \leq_C \gamma(d)$. 
GCs carry both advantages and  drawbacks. One major benefit is the so-called calculational style for defining
abstract operations \cite{cou99,mj08}: if $f:C\ra C$ is any concrete operation involved by some semantic definition (e.g., 
integer addition or multiplication) then a corresponding
correct approximation on $A$ is defined by 
$\alpha\circ f\circ \gamma :A\ra A$, which turns out to be the best possible approximation 
of $f$ on the abstract domain $A$ and, as envisioned by Cousot~\cite{cou99},
allows to systematically derive abstract operations in a correct-by-design manner. On the negative side, GCs have two main weaknesses. 
First, GCs formalize an ideal situation
where each concrete property in $C$ has a unique best abstract approximation in $D$.
Some very useful and largely used abstract domains cannot be defined by a GC, being convex polyhedra a
prominent example of abstract domain where no abstraction map can be defined~\cite{ch78}. 
This problem motivated weaker abstract interpretation frameworks which
only need concretization maps~\cite{cc92}. Secondly, it turns out that abstraction maps of GCs cannot be mechanized~\cite{mon98,pic05}, 
meaning that one cannot use automatic formal proof systems like Coq in order to extract certified algorithms of abstract interpretation,
e.g., based on best correct approximations $\alpha\circ f\circ \gamma$. In other terms, the calculational approach of abstract
interpretation cannot be
automatized. Notably, Verasco~\cite{jou,verasco} (and its precursor described in \cite{blazy13}) 
is a static analyzer for C  which has been formally
designed and verified using the Coq proof assistant, and is based on abstract interpretation using only concretization maps. 
This latter motivation was one starting point of Darais and Van Horn~\cite{dv16} for investigating constructive versions
Galois connections, together with the observation that many useful abstract domains, even if defined by an abstraction
map, still would permit a mechanization of their soundness proofs. Also, Darais and Van Horn's approach \cite{dv16} 
generalizes `Galculator' \cite{so08}, which is a proof assistant based on a given small algebra of Galois connections. 

Constructive Galois connections (CGCs) \cite{dv16} stem from the observation that for many commonly used abstract domains: 
(1) the concrete domain is a powerset (or collecting) domain $\wp(A)$ of an unordered carrier domain $A$; 
(2) the abstraction map $\alpha:\wp(A)\ra D$ is
actually defined as collecting lifting of a basic abstraction function $\eta$ which is defined just on the carrier domain $A$ and takes
values belonging to an unordered abstract domain $B$, 
that is, $\eta:A \ra B$; (3) the concretization map $\mu:B\ra \wp(A)$ provides meaning to basic abstract values ranging in $B$; 
(4) the $\alpha$/$\gamma$ adjunction relation can be equivalently reformulated in terms of the following correspondence between $\eta$ and $\mu$: 
$$
x\in \mu(y) \Lra \eta(x)=y \eqno(\text{CGC-Corr}) 
$$
Moreover, CGCs allow to give 
a soundness condition for pairs of concrete and abstract functions which are defined on carrier concrete and abstract domains $A$ and $B$. 
As a simple example taken from \cite[Section~2]{dv16}, the standard toy parity abstraction for 
integer variables can be defined as a CGC as follows. The carrier concrete domain is $\bZ$, the unordered parity 
domain is
$\bP=\{\even,\odd\}$, abstraction $\parity:\bZ \ra \bP$ and concretization $\mu:\bP \ra \wp(\bZ)$ mappings are straightforwardly defined
and satisfy (CGC-Corr): $z\in \mu(a) \Lra \parity (z) = a$. Also, a successor concrete operation 
$\succe:\bZ \ra \bZ$ is approximated by a sound abstract successor $\succe_\sharp:\bP \ra \bP$ such that 
$\succe_\sharp(\even)=\odd$ and $\succe_\sharp(\odd)=\even$. 

Darais and Van Horn \cite{dv16} also put forward a more general notion of constructive Galois connection for posets (CGP), 
where the carrier concrete domain $A$ and the abstract domain $B$ are posets (rather than unordered sets), and where 
the condition (CGC-Corr) is weakened to: 
$$
x\in \mu(y) \Lra \eta(x) \leq_B y \eqno(\text{CGP-Corr})
$$
This allows them to provide a constructive definition for ordered abstract domains like the sign abstraction of integer variables 
$\Sign = \{\varnothing,\: <\! 0, \: = \!0,\: >\! 0 ,\: \leq\!  0 ,\: 
\neq\! 0 ,\: \geq\! 0 ,\: \bZ \}$ which encodes approximation relations between its abstract values, e.g., the abstract value 
$\geq \!0$ approximates $>\! 0$.

\medskip
\noindent
\textbf{Contributions.}
Our initial observation was that CGCs always encode a partition of the concrete carrier set $A$. As a simple example, for the above 
parity domain $\bP$, the induced partition of $\bZ$ obviously consists of two blocks: $\{z\in \bZ~|~ z ~\text{even}\}$ and 
$\{z\in \bZ~|~ z ~\text{odd}\}$. Furthermore, we also noticed that if an abstract domain $D$ of a collecting domain
$\wp(A)$ does not induce an
underlying partition of $A$ and $D$ is defined through a standard Galois connection $\cG$ then $\cG$ cannot be 
constructively and equivalently formulated by a CGC. Indeed, abstract domains which encode a partition of a given carrier set
have been previously studied and formalized as so-called partitioning Galois connections (PGCs). Intuitively,
a Galois connection defining a domain $D$ which abstracts a concrete collecting domain $\wp(A)$ is called 
partitioning \cite{CC94,rt05} when $D$ represents a partition $\cP$ of the set $A$ together with all the possible unions of blocks
in $\cP$. Our first contribution shows that CGCs are isomorphic to PGCs. This isomorphism is constructive, meaning that 
we define two invertible transforms which map a CGC to an equivalent PGC and vice versa. Moreover, this isomorphism includes 
soundness of
abstract operations, meaning that we also define invertible transforms of pairs of concrete/abstract operations which preserve their
soundness. Secondly, we also investigated Darais and Van Horn's CGPs, in order to characterize them as a suitable class of Galois connections.  
We show that CGPs actually amount to plain GCs of a powerdomain, and therefore CGPs are not able to isolate a specific
class of GCs. This is a negative finding: the generalization from CGCs to CGPs loses the constructive attribute of CGCs. 
Drawing on these results, our third and most significant contribution is the definition of a novel class of
constructive Galois connections, called \emph{purely constructive} GCs (PCGCs), which is more flexible than CGCs while retaining
a constructive approach to Galois connections. The basic idea underlying PCGCs is as follows. 
CGCs essentially represent a partition of the carrier concrete domain
$A$ through an abstract domain $B$. We showed that this partition representation in $B$ implicitly 
brings all the possible unions of its blocks. 
We generalize this approach by allowing to select which unions of blocks to consider in the abstract domain $B$. 
Hence, $B$ may be defined as a partition $\cP$ of $A$ together with an explicit choice of unions of blocks of $\cP$,
where this selection may range from none to all. As an example, a sign abstraction like $\Sign^- \ud \Sign \smallsetminus 
\{\neq\! 0\}$ cannot be formalized as a CGC, although $\Sign^-$ still  represents a partition of $\bZ$ since $\Sign^-$ just lacks
the union of blocks corresponding to the abstract value $\neq\! 0$, that is, the union of $<\! 0$ and $>\! 0$. 
In our framework, $\Sign^-$ can be exactly defined as a PCGC. Moreover, PCGCs come together with a 
definition of sound abstract operations, which also accomodates the standard notions of completeness commonly 
used in abstract interpretation. This paper therefore advocates the use of PPGCs as a
suitable class of constructive abstract domains for a mechanized approach to abstract interpretation. 

\section{Background}
\paragraph*{Notation.}

Let $f:A\ra B$, $g:A\ra \wp(B)$ and $h:\wp(A)\ra B$, $k:A\ra C$, where $A$ and $B$ are sets and 
$C$ is a complete lattice with lub $\vee$. We then use the following definitions:
\[
\begin{array}{rll}
\text{powerset (or collecting) lifting:}& \quad f^\diamond:\wp(A)\ra \wp(B) &f^\diamond (X) \ud \{f(x)~|~x\in X\}\\
\text{singleton powerset lifting:}& \quad f^\srhd: A \ra \wp(B) & f^\srhd (a) \ud \{f(b)\}\\
\text{domain powerset lifting:}& \quad g^*: \wp(A) \ra \wp(B) &g^* (X) \ud \cup_{x\in X} g(x)\\
\text{singleton lowering:}&\quad \sing{h}:A\ra B &\sing{h}(a) \ud h(\{a\}) \\
\text{lub domain powerset lifting:}&\quad k^\vee :\wp(A)\ra C &k^\vee (X) \ud \vee_{a\in X} k(a) 
\end{array}
\]
\noindent
Somewhere we use $f(X)$ as an alternative notation for $f^\diamond (X)$.  
If $A$ is a poset and $X\subseteq A$ then ${\downarrow \!X} \ud \{y\in A~|\exists x\in X.y\leq x\}$, and, in turn,
$\wpd(A) \ud \{X\subseteq A~|~X = \:\downarrow \!X\}$ denotes the downward powerdomain of $A$, 
which ordered by subset inclusion, is a complete lattice. We use $\downarrow \!a$ as a shorthand for $\downarrow \!\{a\}$.
Recall that any set $A$ can be viewed as poset w.r.t.\ the so-called discrete partial order $\leq$: 
for all $x,y\in A$, $x\leq y$ iff $x=y$. 
Let us also recall that $\cP\subseteq \wp(A)$ is a partition of $A$ when: $B\in \cP \Ra B\neq \varnothing$; if $B_1,B_2\in \cP$ and $B_1\neq B_2$ then
$B_1\cap B_2 = \varnothing$; $\cup_{B\in \cP} B = A$. 

\paragraph*{Galois connections.}
Recall that $\cG=\tuple{\alpha,C,D,\gamma}$ is a Galois connection (GC) 
when $C$ and $D$ are posets, $\alpha:C\ra D$, $\gamma:D\ra C$ and
$\alpha(c) \leq_D d\:\Lra
c\leq_C \gamma(d)$. By following a standard terminology in abstract interpretation, $C$ and $D$ are 
called  concrete and abstract domains, while $\alpha$ and $\gamma$ are called abstraction and concretization maps. 
$\cG$ is a disjunctive GC when $\gamma$ is additive (intuitively meaning that it abstractly 
represents logical disjunctions with no loss of precision). 
$\cG$ is a Galois insertion (GI) when $\alpha$ is surjective (or, equivalently, $\gamma$ is injective).

Let us recall some standard definitions and terminology of abstract interpretation \cite{CC77,CC79}.
Let $f:C\ra C$ and $f_\sharp:D\ra D$ be, respectively, concrete and abstract functions. 
The pair $\tuple{f,f_\sharp}_\cG$ is sound (w.r.t.\ $\cG$) when $\alpha\circ f\circ \gamma \sqsubseteq f_\sharp$, 
or, equivalently, $\alpha \circ f \sqsubseteq f_\sharp \circ \alpha$. Also, the pair $\tuple{f,f_\sharp}_\cG$ is optimal when $\alpha\circ f\circ \gamma = f_\sharp$,
backward complete when $\alpha \circ f = f_\sharp \circ \alpha$, forward complete when 
$f \circ \gamma = \gamma \circ f_\sharp$, precise when $f = \gamma\circ f_\sharp\circ \alpha$. 
The abstract function 
$f_\cG\ud \alpha\circ f\circ \gamma$ is called the best correct approximation (BCA) of $f$ induced by $\cG$.

Let $\cG_1 = \tuple{\alpha_1,C,D_1,\gamma_1}$ and $\cG_2=\tuple{\alpha_2,C,D_2,\gamma_2}$ be two GCs with a common
concrete domain $C$. $\cG_1$ is more precise than $\cG_2$, 
denoted by $\cG_1 \sqsubseteq \cG_2$, when $\gamma_2(\alpha_2(C)) \subseteq \gamma_1(\alpha_1(C))$. In turn, 
$\cG_1$ and $\cG_2$ 
are isomorphic when $\cG_1 \sqsubseteq \cG_2$ and $\cG_2 \sqsubseteq \cG_1$, i.e., when 
$\gamma_1(\alpha_1(C)) = \gamma_2(\alpha_2(C))$ holds. The intuition is that $\cG_1$ and $\cG_2$ 
abstractly encode the same properties of $C$
up to a renaming of the abstract values in $D_i$. 
If $f^\sharp_1:D_1 \ra D_1$ and $f^\sharp_2:D_2 \ra D_2$ are two abstract functions for $f:C\ra C$ then 
$f^\sharp_1$ is called isomorphic to $f^\sharp_2$ when $\gamma_1 \circ f_1^\sharp \circ \alpha_1 = \gamma_2 \circ f_2^\sharp \circ \alpha_2$.

\section{Constructive Galois Connections}
Constructive Galois connections (CGCs) have been defined 
by Darais and Van Horn \cite[Section~3]{dv16} to provide a Galois connection-like correspondence between 
sets rather than posets: 
$\tuple{\eta,A,B,\mu}_{\text{CGC}}$ is a CGC when $A$ and $B$ are sets, and
$\eta: A\ra B$, $\mu: B\ra \wp(A)$ satisfy the following equivalence 
$$x\in \mu(y) \:\Lra\: \eta(x)=y \eqno(\text{CGC-Corr})$$ 
The intuition is that $A$ is a carrier set of the concrete powerset domain, 
$B$ is an unordered abstract domain, $\eta$ is a representation function for concrete singletons $\{a\}$ 
while $\mu$ is a concretization function, which give rise to a sort of unordered adjunction relation. 
CGCs have the following properties. 

\begin{restatable}[CGC properties]{lemma}{rCGCprop}\label{CGCprop}
Consider a CGC $\tuple{\eta,A,B,\mu}$. 
\begin{itemize}
\item[{\rm (1)}] $\eta(a_1)=\eta(a_2)\:\Lra\: \mu(\eta(a_1))= \mu(\eta(a_2))\:\Lra\: \mu(\eta(a_1))\cap \mu(\eta(a_2)) \neq \varnothing$
\item[{\rm (2)}] $\mu(b)=\varnothing \:\Lra\: b\not\in \eta(A)$ 
\end{itemize}
\end{restatable}

\medskip
\noindent
As a consequence, we have that $\{\mu(\eta(a))\}_{a\in A}$ are the blocks of 
a partition of $A$, because $A = \cup_{a\in A} \mu(\eta(a))$ and
if $\mu(\eta(a_1))\neq \mu(\eta(a_2))$ then $\mu(\eta(a_1))\cap \mu(\eta(a_2)) =\varnothing$. 

Darais and Van Horn \cite[Section~3.1]{dv16} also define constructive Galois connections for posets (CGPs) as follows. 
$\tuple{\eta,A,B,\mu}$ is a CGP when $\tuple{A,\leq_A}$ and $\tuple{B,\leq_B}$ 
are posets (so that $\wpd(A)_\subseteq$ is a complete lattice), 
$\eta: A\ra B$ and $\mu: B\ra \wpd(A)$ are monotone and the following equivalence holds:
$$x\in \mu(y) \:\Lra\: \eta(x)\leq_B y\eqno(\text{CGP-Corr})$$
Hence, in CGP-Corr $\leq_B$ replaces $=$ of CGC-Corr. We focus on the following properties of CGPs.

\begin{restatable}[CGP properties]{lemma}{rCGPprop}\label{CGP-prop}
Consider a CGP $\tuple{\eta,A,B,\mu}$. 
\begin{itemize}
\item[{\rm (1)}] $\eta(a_1)= \eta(a_2) \:\Lra\: \mu(\eta(a_1))=\mu(\eta(a_2))$
\item[{\rm (2)}] $\mu(b) = \varnothing \:\Lra \:\downarrow\!\! b \cap \eta(A) = \varnothing$
\item[{\rm (3)}] If $B$ is a complete lattice then $\tuple{\eta^\vee,\wpd(A),B,\mu}$ is a GC
\item[{\rm (4)}] $\mu(B) = \mu(\etav(\wpd(A)))$ 
\end{itemize}
\end{restatable}

\begin{example}\label{exCGC}
\rm
Consider $\mathbb{Z}$ with the discrete partial order, so that $\wpd(\bZ)=\wp(\bZ)$, 
$B\ud \{+,\top \}$ with ordering $+ \leq \top$,
$\eta:\bZ\ra B$ be defined by $\eta(x) \ud \textbf{if} \:x>0\: \textbf{then}\, +\, \textbf{else}\, \top$ and $\mu: \bZ\ra \wp(\bZ)$ be defined
by $\mu(+)=\bZ_{>0}$ and $\mu(\top)=\bZ$. It turns out that $\tuple{\eta,\bZ,B,\mu}$ is not a CGC, because $1\in \mu(\top)$ 
while $+ = \eta(1) \neq \top$. Instead, $+=\eta(1) \leq \top$ holds, and indeed this is a CGP. Notice that $\{\mu(\eta(z))~|~z\in \bZ\} =
\{\bZ_{>0},\bZ\}$ is not a partition of $\bZ$. Besides, 
if $B'=\{-,0,+,\bot\}$, $\beta:\bZ\ra B'$ encodes the sign of an integer, and $\delta:B'\ra \wp(\bZ)$ is defined by:
$\delta(-)=\bZ_{<0}$, $\delta(0)=\{0\}$, $\delta(+)=\bZ_{>0}$, $\delta(\bot)=\varnothing$, then $\tuple{\beta,\bZ,B',\delta}$ is clearly
a CGC.
\qed
\end{example}

In the following we will need to compare CGCs with a common concrete carrier set.  
Thus, consider two 
CGCs $\cC_1 = \tuple{\eta_1,A,B_1,\mu_1}$ and $\cC_2 = \tuple{\eta_2,A,B_2,\mu_2}$. 
Then, $\cC_1$ is defined to be \emph{more precise than} $\cC_2$ (or, $\cC_2$ is more abstract than $\cC_1$) 
when $\mu_2(B_2) \subseteq \mu_1 (B_1)$, and this is denoted by $\cC_1 \sqsubseteq \cC_2$. 
Also, $\cC_1$ and $\cC_2$ are 
\emph{isomorphic} when $\cC_1 \sqsubseteq \cC_2$ and $\cC_2 \sqsubseteq \cC_1$, i.e., when
$\mu_1 (B_1) = \mu_2(B_2)$, and this is denoted by
$\cC_1 \cong \cC_2$. The intuition is that two CGCs
are isomorphic when they represent the same abstraction $\mu_i(B_i)$ of $\wp(A)$ 
up to a renaming of the abstract values. 
This notion of isomorphism is justified by the following  result, where $f_{1,2}$ and $f_{2,1}$ play the role
of renaming functions for abstract values. 

\begin{restatable}[CGC Isomorphism]{lemma}{rCGCproptwo}\label{CGCprop2}
Consider $\cC_1=\tuple{\eta_1,A,B_1,\mu_1}$ and $\cC_2=\tuple{\eta_2,A,B_2,\mu_2}$ CGCs. Then,
$\cC_1 \cong \cC_2$ iff there exist $f_{1,2}:\eta_1 (A) \ra \eta_2(A)$ and  
$f_{2,1}:\eta_2 (A) \ra \eta_1(A)$
such that $f_{1,2} \circ f_{2,1} = \id =f_{2,1} \circ f_{1,2}$, $\mu_1 \circ \eta_1 = \mu_2 \circ f_{1,2} \circ \eta_1$
and $\mu_2 \circ \eta_2 = \mu_1 \circ f_{2,1} \circ \eta_2$.
\end{restatable}

We also define a notion of \emph{nonempty isomorphism} $\cong_\varnothing$
which does not take into account possible empty sets in $\mu(B_i)$: 
$\tuple{\eta_1,A,B_1,\mu_1} \cong_\varnothing \tuple{\eta_2,A,B_2,\mu_2}$ when 
$\mu_1 (B_1)\cup \{\varnothing\} = \mu_2(B_2)\cup \{\varnothing\}$. This is justified by the observation
that for any CGC $\tuple{\eta,A,B,\mu}$, we have that 
$\tuple{\eta,A,B,\mu} \cong_\varnothing \tuple{\eta,A,\eta(A),\mu}$, because, by Lemma~\ref{CGP-prop}~(2),
all the abstract values in $B\smallsetminus \eta(A)$ represent the empty set. 
These can therefore be viewed as ``useless'' abstract values and lead to a notion of \emph{constructive Galois
insertion} (CGI) which is the analogue of GI: $\tuple{\eta,A,B,\mu}$ is a CGI when it is a CGC and $\eta$ is surjective.

\section{Partitioning Galois Connections}
Partitioning Galois connections/insertions (PGCs/PGIs) 
have been introduced in \cite[Section~5]{CC94}: given a partition $\cP$ of a set $A$, 
any subset $X\in \wp(A)$ is over-approximated by the unique minimal cover of $X$ through blocks in $\cP$. PGCs have
been studied and used in \cite{rt05,rt07} for generalizing strong preservation of temporal logics in model checking. 
Let $\cG=\tuple{\alpha,\wp(A)_\subseteq,D_\leq,\gamma}$ be a Galois connection, 
where $A$ is any carrier set and $D$ is a poset, and let $\prt(\cG)\ud \{\gamma(\alpha(\{a\}))\}_{a\in A}$. 
$\cG$ is called a
\emph{partitioning Galois connection} when: (1)~$\prt(\cG)$ 
is a partition of $A$; (2)~$\gamma$ is additive, i.e., $\gamma$ preserves arbitrary lub's. 
The main feature of a PGC is that any abstract value $d$ represents a union of blocks of the partition $\prt(\cG)$, 
namely $\gamma(d) = \cup_{a\in \gamma(d)} \gamma(\alpha(\{a\}))$, and, vice versa, for any set of blocks 
$\{\gamma(\alpha(\{a\}))~|~ a\in S\}$ of the partition $\prt(\cG)$, 
for some $S\in\wp(A)$, there exists $d\in A$ such that $\gamma(d)
= \cup \{\gamma(\alpha(\{a\}))~|~ a\in S\}$, where $d=\alpha(S)$. 
In other terms, the abstract domain $D$ is a representation of 
all the possible unions of blocks in $\prt(\cG)$. Although this is the standard definition of PGC, instead of representing 
all the possible unions of blocks of a partition, one could equivalently represent no union of blocks at all: 
this means that the condition~(2) of PGCs of having an additive concretization map $\gamma$ could be replaced by $(2')$: 
if $x,y\in D$ and $x,y$ are uncomparable then $\gamma(x \vee_D y)=A$.  Hence, if  
$\alpha(\{a_1\})$ and $\alpha(\{a_2\})$ represent in $D$ two different blocks then their lub represent no information at all 
(i.e., $\gamma(\alpha(\{a_1,a_2\}))=A$).

It turns out that the notion of CGC is completely equivalent to that of PGC, in the following precise meaning 
based on a pair of invertible transforms of CGCs and PGCs.   
\begin{restatable}[CGC-PGC Equivalence]{theorem}{rteoPGC}\label{teoPGC} \ \ 
\begin{itemize}
\item[{\rm (1)\/}] If 
$\cC = \tuple{\eta,A,B,\mu}$ is a CGC then {\rm $\bT_{\text{PGC}}(\cC)$} $\ud 
\tuple{\eta^\diamond,\wp(A)_\subseteq,\wp(B)_\subseteq,\mu^*}$ is a PGC.
\item[{\rm (2)\/}] If $\cG=\tuple{\alpha,\wp(A)_\subseteq,D_\leq,\gamma}$ is a PGC then 
{\rm $\bT_{\text{CGC}}(\cG)$} $\ud \tuple{\sing{\alpha},A,\{\alpha(\{a\})~|~a\in A\},\gamma}$
is a CGC. 
\item[{\rm (3)\/}] The transforms {\rm $\bT_{\text{PGC}}$} and {\rm $\bT_{\text{CGC}}$} 
are one the inverse of the other, up to nonempty 
isomorphism. 
\end{itemize}
\end{restatable}

\begin{example}\label{exsign}\rm 
Consider the standard abstract domain for sign analysis as encoded by the GI 
$\cS=\tuple{\alpha,\wp(\bZ),\Sign,\gamma}$, where $\Sign$ is the following lattice:

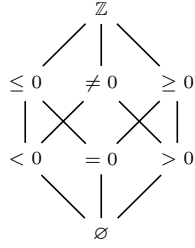
{\centering{
    \begin{tikzpicture}[scale=0.5]
      \draw (0,6) node[name=l] {{\scriptsize$\bZ$}};
      \draw (2,4) node[name=k] {{\scriptsize $\geq 0$}};
      \draw (0,4) node[name=j] {{\scriptsize $\neq 0$}};
      \draw (-2,4) node[name=i] {{\scriptsize $\leq 0$}};
      \draw (2,2) node[name=h] {{\scriptsize $>0$}};
      \draw (-2,2) node[name=g] {{\scriptsize $<0$}};
      \draw (0,2) node[name=d] {{\scriptsize $=0$}};
      \draw (0,0) node[name=a] {{\scriptsize$\varnothing$}};
 
      \draw[semithick] (i) -- (g);
      \draw[semithick] (i) -- (d);

      \draw[semithick] (j) -- (g);
      \draw[semithick] (j) -- (h);

      \draw[semithick] (k) -- (h);
      \draw[semithick] (k) -- (d);

      \draw[semithick] (l) -- (i);
      \draw[semithick] (l) -- (j);
      \draw[semithick] (l) -- (k);
      \draw[semithick] (a) -- (d);
      \draw[semithick] (a) -- (h);
      \draw[semithick] (a) -- (g);
\end{tikzpicture}
\par
} }

\noindent
while abstraction and concretization maps are defined as usual. Let us observe that $\cS$ is indeed
a PGC (more precisely, a PGI), where $\prt(\cS)=\{\bZ_{<0},\bZ_{=0},\bZ_{>0}\}$. 
Then, $\bT_{\text{CGC}}(\cS)$ provides a CGC which is nonempty isomorphic to the
CGC $\cC=\tuple{\beta,\bZ,B',\delta}$ of Example~\ref{exCGC}: these two CGCs only differ 
for the element $\bot\in B'$ whose meaning is $\varnothing=\delta(\bot)$. 
Conversely, $\bT_{\text{PGC}}(\cC)$ 
is a PGC which is isomorphic to the PGI $\cS$: $\wp(B')$ is the abstract domain of $\bT_{\text{PGC}}(\cC)$, so that,
since $B'$ includes the ``useless'' value $\bot$, we obtain a PGC rather than a PGI, because its concretization map $\delta^{*}$
is not injective, e.g., $\delta^*(\{\bot,+\})=\delta^*(\{\bot\})$.  
\qed
\end{example}

Furthermore, it turns out that the CGC/PGC transforms of Theorem~\ref{teoPGC} preserve the 
relative precision relations as follows.  
\begin{restatable}{corollary}{rcoroPGC}\label{coroPGC}
If $\cC_1$ and $\cC_2$ are CGCs  then 
$\cC_1 \sqsubseteq \cC_2$ iff {\rm $\bT_{\text{PGC}}(\cC_1)$} $\sqsubseteq$ {\rm $\bT_{\text{PGC}}(\cC_2)$}. 
\end{restatable}
As a consequence, one can define a lattice of CGCs, ordered w.r.t.\ their relative precision
up to isomorphism, which is order-theoretically isomorphic to the 
so-called lattice of partitioning abstract domains \cite[Theorem~3.2]{rt07}. 

Let us also recall that \cite{dv16} also put forward a definition of 
Kleisli Galois connection (KGC) between posets, which relies on a ``monadic'' notion
of abstraction/concretization maps. Actually,  this class of constructive abstractions is shown
in \cite[Section~6]{dv16} 
to be equivalent to CGCs, where this isomorphism includes the notion of soundness (and optimality) for
abstract functions. Hence, we do not need to replicate our isomorphism between KGCs and PGCs, which automatically
follows.

\paragraph*{CGCs as Least Disjunctive Bases.}
Given a CGC
$\cC = \tuple{\eta,A,B,\mu}$, 
Theorem~\ref{teoPGC} shows that
$\bT_{\text{PGC}}(\cC) =
\tuple{\eta^\diamond,\wp(A)_\subseteq,\wp(B)_\subseteq,\mu^*}$ is a PGC.
We observe that $\{\{x\} ~|~x\in B\}$ is the set of join-irreducible elements of the complete lattice $\wp(B)_\subseteq$~---~recall 
that  an element $x$ of a complete lattice is join-irreducible when, for any $S$, $x=\vee S \Ra x\in S$. 
In abstract interpretation terms \cite{gr98}, this observation 
means that $\{\{x\} ~|~x\in B\}$ can be essentially viewed 
as the so-called least disjunctive basis of the partitioning abstract domain $\wp(B)_\subseteq$. 
Least disjunctive basis have been introduced in \cite{gr98} as an inverse operation
to the disjunctive completion of abstract domains, that is, the least disjunctive refinement of an abstract domain $D$. 
Given an abstract domain $D$, its least disjunctive basis
is defined to be the most abstract domain having the same disjunctive completion as $D$. Hence, the least disjunctive basis
of $D$ reveals and therefore removes all the disjunctive information inside $D$. 
A concrete domain which is a powerset (ordered by subset inclusion) satisfies the hypotheses 
of \cite[Theorem~4.13]{gr98}, so that the least disjuctive basis of an abstract domain $D$ exists and is
characterized as the closure under arbitrary meets of the join-irreducible elements of $D$. This result can be applied  
to the abstract domain $\wp(B)_\subseteq$ of the PGC $\bT_{\text{PGC}}(\cC)$, whose least disjunctive basis 
is given by the meet-closure of $\{\{x\} ~|~x\in B\}$, where this meet-closure simply adds $\varnothing$ and $B$.  
In this sense, this section systematically reconstructed the notion of CGC as least disjunctive basis of a partitioning abstract domain.

\paragraph*{Constructive Closure Operators.}
In abstract interpretation, abstract domains up to renaming of
abstract values are encoded by closure operators, which are isomorphic to GCs \cite{CC79}. 
Hence, the isomorphism between CGCs and PGCs in Theorem~\ref{teoPGC} leads to the 
following notion. Given any set $A$,  
a map $\varphi: A \ra \wp(A)$ is a \emph{constructive closure operator} (CCO) when the following condition holds:
$x\in \varphi(y) \:\Lra\: \varphi(x) =\varphi(y)$.

\noindent
CCOs turn out to be the closure operator counterpart of CGCs, as shown by the following result. 
\begin{restatable}[CGC-CCO Equivalence]{corollary}{rteoCCO}\label{teoCCO}\ 
\begin{itemize}
\item[{\rm (1)\/}] If $\cC=\tuple{\eta,A,B,\mu}$ is a CGC then {\rm $\bT_{\text{CCO}}$} $(\cC)\ud \mu\circ \eta: A\ra \wp(A)$ is a CCO.
\item[{\rm (2)\/}] If $\varphi: A \ra \wp(A)$ is a CCO then {\rm $\bT_{\text{CGC}}$} $(\varphi)\ud \tuple{\varphi,A,\{\varphi(a)~|~a\in A\},\id}$ is a CGC.
\item[{\rm (3)\/}] The transforms {\rm $\bT_{\text{CCO}}$} and {\rm $\bT_{\text{CGC}}$} are one the inverse of the other, up to nonempty 
isomorphism. 
\end{itemize}
\end{restatable}

\subsection{Characterization of CGPs}
Let us now focus on CGPs. Can this class of constructive abstractions be characterized in terms of Galois connections? 
Let the carrier concrete set $A$ be a poset, 
the abstract domain $B$ be a complete lattice, 
$\eta: A\ra B$ and $\mu:B\ra \wpd(A)$ be monotone.  
By relying on CGPs/GCs transforms, we show that the class of CGPs turns out to be isomorphic to the class of GCs
of the concrete powerdomain $\wpd(A)$. 
\begin{restatable}[CGP-GC Equivalence]{theorem}{rteoCGP}\label{teoCGP}
\ 
\begin{itemize}
\item[{\rm (1)\/}] If
$\cC=\tuple{\eta,A,B,\mu}$ is a CGP then
{\rm $\bT_{\text{GC}}(\cC) \ud \tuple{\eta^\vee,\wpd(A)_\subseteq,B, \mu}$} is a GC.
\item[{\rm (2)\/}] If $\cG=\tuple{\alpha,\wpd(A)_\subseteq,D_\leq,\gamma}$ is a GC then 
{\rm $\bT_{\text{CGP}}(\cG) \ud\tuple{\lambda a.\alpha(\downarrow\!\!\{a\}),A,D,\gamma}$}
is a CGP. 
\item[{\rm (3)\/}] The transforms {\rm $\bT_{\text{GC}}$} and {\rm $\bT_{\text{CGP}}$} 
are one the inverse of the other, up to isomorphism between GCs. 
\end{itemize}
\end{restatable}
Otherwise stated, this result shows that CGPs, specifically defined for adapting 
CGCs to carrier sets which are posets, indeed boil down to plain GCs of the powerdomain $\wpd(A)$, and 
therefore lose their constructive attribute.  

\begin{example}\rm 
Consider the following lattice $D$ of integer intervals ordered by subset inclusion:

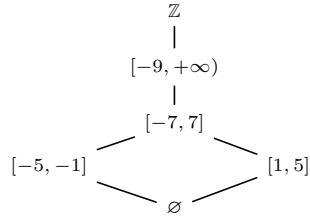
{\centering{
    \begin{tikzpicture}[scale=0.75]
      \draw (0,3.5) node[name=j] {{\scriptsize$\bZ$}};
      \draw (0,2.5) node[name=i] {{\scriptsize$[-9,+\infty)$}}; 
      \draw (0,1.5) node[name=g] {{\scriptsize$[-7,7]$}};
      \draw (2,0.75) node[name=e] {{\scriptsize$[1,5]$}};   
      \draw (-2.2,0.75) node[name=c] {{\scriptsize$[-5,-1]$}};      
      \draw (0,0) node[name=a] {{\scriptsize$\varnothing$}};

      \draw[semithick] (a) -- (c);
      \draw[semithick] (a) -- (e);
      \draw[semithick] (g) -- (c);
      \draw[semithick] (g) -- (e);
      \draw[semithick] (g) -- (i);
      \draw[semithick] (j) -- (i);

\end{tikzpicture}
\par
}
}

\noindent
By considering $\bZ$ as discretely ordered, this lattice $D$ gives rise to a GI $\cG=\tuple{\alpha,\wp(\bZ),D,\gamma}$ 
where $\gamma$ is the identity and, for example, we have that $\alpha(\{2\})=[1,5]$, $\alpha(\{0\})=\alpha(\{6\})=[-7,7]$, $\alpha(\{10\}) = [-9,+\infty)$, 
$\alpha(\{-10\})=\bZ$. Then, by Theorem~\ref{teoCGP}~(2), 
$\bT_{\text{CGP}}(\cG) =\tuple{\lambda z.\alpha(\{z\}),\bZ,D,\gamma}$ is a CGP. Ler us remark that the GI $\cG$ 
is neither partitioning
nor disjunctive, so that the intuition is that the CGP $\bT_{\text{CGP}}(\cG)$ should not be considered as being ``constructive''. 
As a limit infinite example, consider the complete lattice $E\ud \{[0,n] ~|~ n\in \bN\} \cup \{\bN\}$, ordered
by subset inclusion, which is an infinite increasing chain of intervals of natural numbers. This complete lattice gives rise to a 
GI $\cE = \tuple{\alpha,\wp(\bN),E,\gamma}$ where $\bN$ is discretely ordered and
$\gamma$ is the identity. Here, Theorem~\ref{teoCGP}~(2) yields a legal CGP
$\bT_{\text{CGP}}(\cE) = \tuple{\eta, \bN,E,\id}$ where $\eta(n)=[0,n]$ and whose constructive trait can hardly be 
identifiable.
\qed
\end{example}

\section{Soundness of Abstract Operations}
Our next objective is to transform a sound pair of functions from CGCs to PGCs and vice versa, in order to show 
that the equivalence between CGCs and PGCs also include soudness (and further optimality conditions) 
of abstract functions.  For notational simplicity, we consider unary functions, but the whole approach can be
straighforwardly generalized to generic $n$-ary functions (that we will use in some examples).

Let $\cC = \tuple{\eta,A,B,\mu}$ be a CGC, $f:A\ra A$ be a concrete function and $f_\sharp:B\ra B$ be
a corresponding abstract function. Let us recall that Darais and Van Horn~\cite{dv16} 
provide four equivalent soudness conditions for $\tuple{f,f_\sharp}$ to be sound w.r.t.\ $\cC$:  
\begin{align*}
& x\in \mu(y) \:\&\: y'=\eta(f(x)) \;\Ra\; y'=f_\sharp(y) &  \text{(CGC-Snd/$\eta\mu$)}&\\
& x\in \mu(y) \:\&\: x'=f(x) \;\Ra\; x'\in \mu(f_\sharp(y)) & \text{(CGC-Snd/$\mu\mu$)}&\\
& y=\eta(f(x)) \;\Ra\; y=f_\sharp(\eta(x)) &  \text{(CGC-Snd/$\eta\eta$)} &\\
& x'= f(x) \;\Ra\; x'\in \mu (f_\sharp(\eta(x))) & \text{(CGC-Snd/$\mu\eta$)}&
\end{align*}

On the one hand, on the transformed 
PGC $\bT_{\text{PGC}}(\cC) = \tuple{\eta^\diamond,\wp(A)_\subseteq,\wp(B)_\subseteq,\mu^*}$,
we simply the powerset lifting of concrete and abstract functions, namely, we define 
$\bT_{\text{PGC}}(\tuple{f,f_\sharp})\ud \tuple{f^\diamond,f_\sharp^\diamond}$, where
$f^\diamond:\wp(A)\ra \wp(A)$ and $f_\sharp^\diamond:
\wp(B)\ra \wp(B)$. 

On the other hand, let $\cG = \tuple{\alpha,\wp(A)_\subseteq,D_\leq,\gamma}$ be a PGC. On PGCs, we consider
concrete functions on $\wp(A)$ which are defined as 
collecting version of a mapping $g:A\ra A$ on the carrier set $A$, that is, $g^\diamond:\wp(A)\ra\wp(A)$ will be our
concrete function. 
A corresponding monotone 
abstract function $g_\sharp: D\ra D$ is called 
\emph{block-preserving} when $g_\sharp$ maps (abstract representations of) blocks 
to (abstract representations of) blocks, namely, when the following condition holds: 
$\forall a\in A.\exists a'\in A.\: g_\sharp (\alpha(\{a\})) = \alpha(\{a'\})$.

\begin{restatable}{lemma}{rlemmaone}\label{lemma-one}
If $\cG$ is a PGI, $\tuple{g^\diamond,g_\sharp}$ is sound and $g_\sharp$ is block-preserving then, for any $a\in A$, 
$g_\sharp(\alpha(\{a\})) = \alpha(\{g(a)\})$ and
$g^\diamond (\gamma(\alpha(\{a\}))) \subseteq \gamma(\alpha(\{g(a)\}))$.
\end{restatable}

Thus, given a 
sound pair $\tuple{g^\diamond,g_\sharp}$ w.r.t.\ $\cG$ where $g_\sharp$ is block preserving,
on the transformed CGC $\bT_{\text{CGC}}(\cG) = \tuple{\sing{\alpha},A,\{\alpha(\{a\})~|~a\in A\},\gamma}$
we 
consider the carrier concrete function $g:A\ra A$ and the following restriction of the abstract function $g_\sharp$ to
abstract representations of blocks:  
$g_\sharp^r: \{\alpha(\{a\})~|~a\in A\} \ra \{\alpha(\{a\})~|~a\in A\}$ which, by Lemma~\ref{lemma-one}, can be defined  
as $g_\sharp^r(\alpha(\{a\})) \ud \alpha(\{g(a)\})$. We denote this transform by 
$\bT_{\text{CGC}}(\tuple{g^\diamond,g_\sharp})\ud \tuple{g,g_\sharp^r}$.

Given two CGCs $\cC_i = \tuple{\eta_i,A,B_i,\mu_i}$, $i=1,2$, a concrete function $f:A\ra A$ and two corresponding
abstract functions $f^\sharp_i:B_i \ra B_i$, we use the notation 
$\tuple{f,f_1^\sharp} \cong \tuple{f,f_2^\sharp}$ to mean that $f_i$ is sound for $f$ w.r.t.\
$\cC_i$ and
$\mu_1 \circ f_1^\sharp \circ \eta_1 = \mu_2 \circ f_2^\sharp \circ \eta_2$, that is,  
the concrete projections of $f_1^\sharp$ and $f_2^\sharp$ coincide so that these functions can be viewed as
isomorphic.  With this notion, our correspondance between CGCs and PGCs can be extended to include soundness
as follows.

\begin{restatable}{theorem}{rtheosound}\label{theosound}\ 
\begin{itemize}
\item[{\rm (1)\/}] Let $\cC = \tuple{\eta,A,B,\mu}$ be a CGC, $f:A\ra A$ and $f_\sharp: B\ra B$. Then, 
$\tuple{f,f_\sharp}$ is sound iff {\rm $\bT_{\text{PGC}}(\tuple{f,f_\sharp})$} is sound
w.r.t.\ {\rm $\bT_{\text{PGC}}(\cC)$}. 
\item[{\rm (2)\/}] Let $\cG = \tuple{\alpha,\wp(A)_\subseteq,D_\leq,\gamma}$ be a PGC, 
$g^\diamond:\wp(A)\ra \wp(A)$, for some $g:A\ra A$,  and 
$g_\sharp:D\ra D$ be monotone and block-preserving. Then,
$\tuple{g^\diamond,g_\sharp}$ is sound iff {\rm $\bT_{\text{CGC}}(\tuple{g^\diamond,g_\sharp})$} is sound
w.r.t.\ 
{\rm $\bT_{\text{CGC}}(\cG)$}. 
\item[{\rm (3)\/}] If $\tuple{f,f_\sharp}$ is sound then {\rm $\bT_{\text{CGC}}(\bT_{\text{PGC}}(\tuple{f,f_\sharp}))$} 
$\cong \tuple{f,f_\sharp}$.
If $\tuple{g^\diamond,g_\sharp}$ is sound and $g_\sharp$ is block-preserving and additive then 
{\rm $\bT_{\text{PGC}}(\bT_{\text{CGC}}(\tuple{g^\diamond,g_\sharp}))$} $\cong \tuple{g^\diamond,g_\sharp}$.
\end{itemize}
\end{restatable}

\begin{example}\rm 
Consider the PGI $\cS=\tuple{\alpha,\wp(\bZ),\Sign,\gamma}$ for the standard sign domain introduced in 
Example~\ref{exsign}. Consider the square operation $\sq:\bZ\ra \bZ$ such that $\sq(z)=z^2$ and its  collecting
lifting $\sq^\diamond: \wp(\bZ) \ra \wp(\bZ)$. Correspondingly, consider $\sq_\cS: \Sign \ra \Sign$ defined
as BCA of $\sq^\diamond$, namely:
\[
\sq_\cS = \{\varnothing \mapsto \varnothing,\: <\! 0 \mapsto >\! 0,\: =\! 0 \mapsto\: =\! 0,\: >\! 0 \mapsto >\! 0,\: \leq\!  0 \mapsto \geq\! 0,\: 
\neq\! 0 \mapsto >\! 0,\: 
\geq\! 0 \mapsto \geq\! 0,\: \bZ \mapsto \geq\! 0\}.
\]
Then, let us observe that $\sq_\cS$ is (monotone and) block-preserving: indeed, the set of (abstract) blocks is 
$B=\{ <\! 0,\, =\! 0,\, >\! 0\}$ and $\sq_\cS$ maps blocks to blocks. 
Hence, here we have that
$\bT_{\text{CGC}}(\cS)=\tuple{\eta,\bZ,B,\mu}$ and $\bT_{\text{CGC}}(\tuple{\sq^\diamond,\sq_\cS})=\tuple{\sq,\sq_\cS^r}$ where
$\sq_\cS^r:\bZ\ra \bZ$ is such that
$\sq_\cS^r (\alpha(\{z\})) = \alpha(\{\sq(z)\})$. 
\qed
\end{example}

\paragraph*{Completeness.}
As observed in \cite{dv16}, the above four equivalent soundness conditions for CGCs 
lead to four non-equivalent conditions of completeness for abstract functions, where $\Lra$ replaces $\Ra$:
\begin{align*}
& x\in \mu(y) \:\&\: y'=\eta(f(x)) \;\Lra\; y'=f_\sharp(y) &  \text{(CGC-Cmp/$\eta\mu$)}\\
& x\in \mu(y) \:\&\: x'=f(x) \;\Lra\; x'\in \mu(f_\sharp(y)) & \text{(CGC-Cmp/$\mu\mu$)}\\
& y=\eta(f(x)) \;\Lra\; y=f_\sharp(\eta(x)) &  \text{(CGC-Cmp/$\eta\eta$)}\\
& x'= f(x) \;\Lra\; x'\in \mu (f_\sharp(\eta(x))) & \text{(CGC-Cmp/$\mu\eta$)}
\end{align*}
It is worth remarking that these completeness conditions for a pair 
$\tuple{f,f_\sharp}$ can be equivalently stated using well known optimality/completeness conditions 
for Galois connections for the transformed pair $\bT_{\text{PGC}}(\tuple{f,f_\sharp})$. 

\begin{restatable}{lemma}{rlemsound}\label{lemsound} \ 
\begin{itemize}
\item[{\rm (1)\/}] $\tuple{f,f_\sharp}$ satisfies {\rm (CGC-Cmp/$\eta\mu$)} iff 
{\rm $\bT_{\text{PGC}}(\tuple{f,f_\sharp})$} is best correct approximation w.r.t.\ {\rm $\bT_{\text{PGC}}(\cC)$}.
\item[{\rm (2)\/}] $\tuple{f,f_\sharp}$ satisfies {\rm (CGC-Cmp/$\mu\mu$)} iff 
{\rm $\bT_{\text{PGC}}(\tuple{f,f_\sharp})$} is forward complete w.r.t.\ {\rm $\bT_{\text{PGC}}(\cC)$}. 
\item[{\rm (3)\/}] $\tuple{f,f_\sharp}$ satisfies {\rm (CGC-Cmp/$\eta\eta$)} iff 
{\rm $\bT_{\text{PGC}}(\tuple{f,f_\sharp})$} is backward complete w.r.t.\ {\rm $\bT_{\text{PGC}}(\cC)$}.
\item[{\rm (4)\/}] $\tuple{f,f_\sharp}$ satisfies {\rm (CGC-Cmp/$\mu\eta$)} iff 
{\rm $\bT_{\text{PGC}}(\tuple{f,f_\sharp})$} is precise  w.r.t.\ {\rm $\bT_{\text{PGC}}(\cC)$}.
\end{itemize}
\end{restatable}

\section{Purely Partitioning Galois Connections}
Drawing on the above results, 
we define a novel  
class of constructive abstract domains, which we call purely constructive Galois connections (PCGCs). 
The idea is that PCGCs generalize CGCs as follows. 
CGCs essentially represent a partition of the carrier concrete domain
$A$ as an abstract domain $B$. We showed that this partition representation also brings all the possible unions of its blocks. 
The goal is to generalize this approach by allowing to decide which unions of blocks to consider in the abstract domain $B$. 
Hence, $B$ may be defined as a partition $P$ of $A$ together with an explicit selection of unions of blocks of $P$,
where this selection may range from none to all. 

A \emph{purely constructive Galois connection} (PCGC) $\tuple{\eta,A,B,\mu}_{\text{PCGC}}$ 
consists of an unordered concrete carrier 
set $A$ and of an ordered abstract domain which is required to be a poset $\tuple{B,\leq}$, 
together with the
mappings
$\eta: A\ra B$ and  $\mu: B\ra \wp(A)$ which satify the following two conditions:
\begin{align*}
(1)~~ & x\in \mu(\eta(x')) \:\Lra\: \eta(x)=\eta(x')\\
(2)~~ & x\in \mu(y) \:\Lra\: \eta(x) \leq y
\end{align*}
Thus, condition (2) coincides with (CGP-Corr), while condition~(1) amounts to (CGC-Corr) restricted to abstract values
ranging in $\eta(A)$. PCGCs have the following properties. 

\begin{restatable}[PCGC properties]{lemma}{rPCGCprop}\label{PCGC-prop}
Consider a PCGC $\tuple{\eta,A,B_\leq ,\mu}$. 
\begin{itemize}
\item[{\rm (1)\/}] $\eta(a_1)=\eta(a_2)\:\Lra\: \mu(\eta(a_1))= \mu(\eta(a_2)) \:\Lra\: \mu(\eta(a_1))\cap \mu(\eta(a_2)) \neq \varnothing$
\item[{\rm (2)\/}] $\mu(b)=\varnothing \:\Ra\: b\not\in \eta(A)$, while the viceversa does not hold 
\item[{\rm (3)\/}] If $B$ is a complete lattice then $\tuple{\eta^\vee, \wp(A)_\subseteq,B_\leq,\mu}$ is a GC
\end{itemize}
\end{restatable}

\noindent
In particular, let us remark that: by Lemma~\ref{PCGC-prop}~(1), $\{\mu(\eta(a))\}_{a\in A}$ still is a partition of $A$; 
by Lemma~\ref{PCGC-prop}~(2), differently from CGCs, if $b\not\in\eta(A)$ it may happen that $\mu(b)\neq \varnothing$; 
by Lemma~\ref{PCGC-prop}~(3), analogously to CGPs, $\eta^\vee$ and $\mu$ give rise to a GC. 

\begin{example}\label{ex-PCGC}\rm 
Consider the following finite lattice $B$ of integer intervals ordered by subset inclusion:

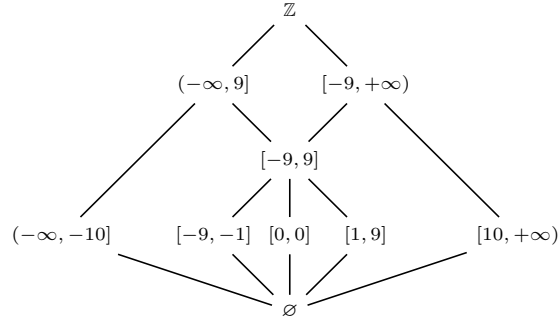
{\centering{
    \begin{tikzpicture}[scale=0.5]
      \draw (0,8) node[name=j] {{\scriptsize$\bZ$}};
      \draw (2,6) node[name=i] {{\scriptsize$[-9,+\infty)$}}; 
      \draw (-2,6) node[name=h] {{\scriptsize$(-\infty,9]$}}; 
      \draw (0,4) node[name=g] {{\scriptsize$[-9,9]$}};
      \draw (6,2) node[name=f] {{\scriptsize$[10,+\infty)$}};
      \draw (2,2) node[name=e] {{\scriptsize$[1,9]$}};   
      \draw (0,2) node[name=d] {{\scriptsize$[0,0]$}};
      \draw (-2,2) node[name=c] {{\scriptsize$[-9,-1]$}};      
      \draw (-6,2) node[name=b] {{\scriptsize $(-\infty,-10]$}};
      \draw (0,0) node[name=a] {{\scriptsize$\varnothing$}};

      \draw[semithick] (a) -- (b);
      \draw[semithick] (a) -- (c);
      \draw[semithick] (a) -- (d);
      \draw[semithick] (a) -- (e);
      \draw[semithick] (a) -- (f);
      \draw[semithick] (g) -- (c);
      \draw[semithick] (g) -- (d);
      \draw[semithick] (g) -- (e);
      \draw[semithick] (g) -- (h);
      \draw[semithick] (g) -- (i);
      \draw[semithick] (j) -- (h);
      \draw[semithick] (j) -- (i);
      \draw[semithick] (b) -- (h);
      \draw[semithick] (f) -- (i);

\end{tikzpicture}
\par
}
}

\noindent
Let $\eta:\mathbb{Z}\ra B$ be defined as follows:
\[
\eta(x) \ud
\begin{cases}
(-\infty,-10] & \text{if~}x\in (-\infty,-10]\\
[-9,-1] & \text{if~} x \in [-9,-1]\\
[0,0] & \text{if~} x =0\\
[1,9] & \text{if~} x \in [1,9]\\
[10,+\infty) & \text{if~}x \in [10,+\infty)
\end{cases}
\] 
while $\mu:B\ra \wp(\mathbb{Z})$ is simply defined as the identity map. 
Then, it is simple to check that $\cP=\tuple{\eta,\mathbb{Z},B,\mu}$ is a PCGC. 
However, it turns out that 
$\cP$ is not a CGC: in fact, $0\in \mu([-9,9])$ while $\eta(0)=[0,0] \neq [-9,9]$. Also, if $\bZ$ is 
considered as a poset w.r.t.\ the discrete order then $\wpd(\bZ) = \wp(\bZ)$ and
$\eta$ and $\mu$ are monotone, so that, by PCGC~(2), $\cP$ is a CGP. 
Consider now $B' \ud (B\smallsetminus \{[-9,9], (-\infty,9], [-9,+\infty)\}) \cup \{[-10,10]\}$, which  
still is a finite lattice 
where 
$[-10,10]$ is the lub in $B'$ of $\{[-9,-1], [0,0], [1,9]\}$,
although $[-10,10]$ is not the union of these blocks. 
In this case, $\tuple{\eta,\mathbb{Z},B',\mu}$
is not a PCGC, because $10\in \mu([-10,10])$ but $\eta(10)=[10,+\infty) \not\subseteq [-10,10]$, so that 
PCGC~(2) does not hold. Finally, consider the CGC $\cC=\tuple{\eta,\bZ,B,\mu}$ defined in Example~\ref{exCGC}. 
Then, $\cC$ is not a PCGC because $1\in \mu(\eta(0)) =\bZ$ but $+=\eta(1) \neq \eta(0)=\top$.  
\qed
\end{example}

Similarly to Theorems~\ref{teoPGC} and~\ref{teoCGP}, let us now characteriza PCGCs as a class of Galois connections.  
Recall that a GC $\cG=\tuple{\alpha,\wp(A)_\subseteq,D_\leq,\gamma}$ is a PGC when $\prt(\cG)$ is a partition of $A$
and $\gamma$ is additive. By dropping this latter requirement of additivity for $\gamma$, we
define $\cG$ to be a
\emph{purely partitioning} Galois connection (PPGC) just when $\prt(\cG)$ is a partition of $A$. 
The terminology ``purely partitioning'' hints at the property (which is not hard to check) 
that the disjunctive completion of $D$ indeed yields a partitioning
Galois connection. It turns out that this class of GCs characterize PCGCs as follows.

\begin{restatable}[PCGC-PPGC Equivalence]{theorem}{rteoPCGC}\label{teoPCGC}\ 
\begin{itemize}
\item[{\rm (1)\/}] If $B_\leq$ is a complete lattice and 
$\cC = \tuple{\eta,A,B_\leq,\mu}$ is a PCGC then {\rm $\bT_{\text{PPGC}}(\cC)$} $\ud 
\tuple{\eta^\vee,\wp(A)_\subseteq, B_\leq,\mu}$ is a PPGC.
\item[{\rm (2)\/}] If $\cG=\tuple{\alpha,\wp(A)_\subseteq,D_\leq,\gamma}$ is a PPGC then 
{\rm $\bT_{\text{PCGC}}(\cG)$} $\ud \tuple{\alphasing,A,D_\leq,\gamma}$
is a PCGC. 
\item[{\rm (3)\/}] The transforms {\rm $\bT_{\text{PPGC}}$} and {\rm $\bT_{\text{PCGC}}$} 
are one the inverse of the other, up to 
isomorphism.
\end{itemize} 
\end{restatable}

\begin{example}\rm
Consider the PCGC $\cP$ defined in Example~\ref{ex-PCGC}, so that  $\bT_{\text{PPGC}}(\cP) = \tuple{\eta^\vee, \wp(\bZ)_\subseteq,B_\leq,\id}$
is a PPGC where the corresponding partition of $\bZ$ is $P=\{(-\infty,10], [-9,-1],$ $[0,0],[1,9],[10,+\infty)\}$ and
the abstraction map $\eta^\vee$ approximates a set of integers $X\in \wp(\bZ)$ by the least  union of blocks of $P$ which belongs to
$B$: for example, $\eta^\vee(\{1,10\})=[-9,+\infty)$ and $\eta^\vee(\{0,1\})=[-9,9]$. \qed
\end{example}

\paragraph*{CGCs as PCGCs as CGPs.}
It turns out that any CGC is indeed a PCGC, which, in turn, is a CGP. 
Let $\tuple{\eta,A,B,\mu}$ be a CGC. 
Firstly, it is enough to consider $B$ as a poset for the discrete partial order $\leq$, since 
this makes $\tuple{\eta,A,B_\leq,\mu}$ a PCGC. In fact: (1) $a\in \mu(\eta(a'))$ iff, by (CGC-Corr), 
$\eta(a)=\eta(a')$; (2) if $b\in \eta(A)$ then $b=\eta(a')$, for some $a'$, so that,  by (CGC-Corr),
$a\in \mu(b)\Lra \eta(a)=b$, while if $b\not\in \eta(A)$, then, by Lemma~\ref{CGCprop}~(2),  $\mu(b)=\varnothing$. 
Furthermore, any PCGC $\tuple{\eta,A,B_\leq,\mu}$ can be viewed as a CGP simply by making the concrete carrier set 
$A$ a poset for the discrete order, so that $\wpd(A)=\wp(A)$, and $\eta$ becomes trivially monotone as well as
$\mu:B\ra \wp(A)$: in fact, if $b_1\leq b_2$ and $a\in \mu(b_1)$ then $\eta(a) \leq b_1 \leq b_2$, so that $a\in \mu(b_2)$. 

\subsection{Soundness of Abstract Operations}
Let $\cC = \tuple{\eta,A,B_\leq,\mu}$ be a PCGC and $f: A\ra A$ be a concrete function.  
By relying on Theorem~\ref{teoPCGC}~(1), we are able to define the BCA of the lifted function
$f^\diamond:\wp(A)\ra\wp(A)$
w.r.t.\ the PPGC $\tuple{\eta^\vee,\wp(A)_\subseteq, B_\leq,\mu} = \bT_{\text{PPGC}}(\cC)$. This is denoted by
$f_\cC: B\ra B$ and is therefore defined by $f_\cC \ud \eta^\vee \circ f^\diamond \circ \mu$, so that: 
$$f_\cC (b) = \vee\{\eta(f(a))~|~ a\in \mu(b)\}.$$ 
Hence, given an abstract function $f_\sharp:B\ra B$, this BCA leads us to define 
$\tuple{f,f^\sharp}$ to be sound for $\cC$ when $f_\sharp$ is less precise than the BCA, that is, when for any $b\in B$,
$f_\cC (b) \leq f_\sharp (b)$.
This turns out to be equivalent to the following condition: $\tuple{f,f^\sharp}$ is sound w.r.t.\ $\cC$ iff 
$$\eta(a) \leq b \:\Ra\: \eta(f(a)) \leq f_\sharp (b)\eqno(\text{PCGC-Snd})$$  

It is then easy to transform a sound pair of concrete/abstract functions
$\tuple{f,f_\sharp}$ for a PCGC $\cC$ into the pair $\bT_{\text{PPGC}}(\tuple{f,f_\sharp})\ud \tuple{f^\diamond,f_\sharp}$
for the corresponding PPGC $\bT_{\text{PPGC}}(\cC)=\tuple{\eta^\vee,\wp(A)_\subseteq, B_\leq,\mu}$. 
Conversely, if $\cD = \tuple{\alpha,\wp(A)_\subseteq,D_\leq,\gamma}$ is a PPGC and 
$\tuple{g^\diamond,g_\sharp}$ is a sound pair for $\cD$, where  
$g^\diamond:\wp(A)\ra \wp(A)$ for some $g:A\ra A$, then 
$\tuple{g^\diamond,g_\sharp}$ is transformed into $\bT_{\text{PCGC}}(\tuple{g^\diamond,g_\sharp})\ud \tuple{g,g_\sharp}$
relatively to the PCGC $\bT_{\text{PCGC}}(\cD)$. A result analogous to Theorem~\ref{theosound} can then be proved. 

\begin{restatable}{theorem}{rteoPCGCsound}\label{teoPCGCsound}\
\begin{itemize}
\item[{\rm (1)\/}] Let $\cC = \tuple{\eta,A,B_\leq,\mu}$ be a PCGC, with $B$ complete lattice, 
$f:A\ra A$ and $f_\sharp : B\ra B$. 
Then, 
$\tuple{f,f_\sharp}$ is sound iff {\rm $\bT_{\text{PPGC}}(\tuple{f,f_\sharp})$} is sound
w.r.t.\  {\rm $\bT_{\text{PPGC}}(\cC)$}. 
\item[{\rm (2)\/}] Let $\cD = \tuple{\alpha,\wp(A)_\subseteq,D_\leq,\gamma}$ be a PPGC, 
$g^\diamond:\wp(A)\ra \wp(A)$,  for some $g:A\ra A$,  and 
$g_\sharp:D\ra D$. Then,
$\tuple{g^\diamond,g_\sharp}$ is sound iff {\rm $\bT_{\text{PCGC}}(\tuple{g^\diamond,g_\sharp})$} is sound
w.r.t.\ {\rm $\bT_{\text{PCGC}}(\cD)$}. 
\item[{\rm (3)\/}] The transforms {\rm $\bT_{\text{PPGC}}$} and {\rm $\bT_{\text{PCGC}}$} are one the inverse of the other. 
\end{itemize}
\end{restatable}

Since $f_\sharp$ is defined to be sound when $\eta^\vee \circ f^\diamond \circ \mu \leq f_\sharp$, it is then natural to define 
$f_\sharp$ optimal when $\eta^\vee \circ f^\diamond \circ \mu = f_\sharp$, backward complete when 
$\eta^\vee \circ f^\diamond = f_\sharp \circ \eta^\vee$ and forward complete when 
$f^\diamond \circ \mu = \mu \circ f_\sharp$.
In particular, these definitions allow us
to apply the abstraction refinement operators introduced in \cite{grs00} for minimally refining the abstract domain 
$B$ in order to obtain a backward/forward complete abstract function and the technique introduced in \cite{gr14} for
simplifying abstract domains while retaining the optimality of abstract operations.

\subsection{An Example of PCGC}
Consider the following infinite complete lattice $B_\leq$.

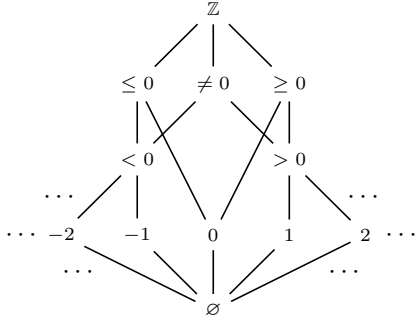
{\centering{
    \begin{tikzpicture}[scale=0.5]
      \draw (0,8) node[name=l] {{\scriptsize$\bZ$}};
      
      \draw (2,6) node[name=k] {{\scriptsize $\geq 0$}};
      \draw (0,6) node[name=j] {{\scriptsize $\neq 0$}};
      \draw (-2,6) node[name=i] {{\scriptsize $\leq 0$}};
      
      \draw (2,4) node[name=h] {{\scriptsize $>0$}};
      \draw (-2,4) node[name=g] {{\scriptsize $<0$}};
      
      \draw (4,2) node[name=f] {{\scriptsize $2$}};
      \draw (2,2) node[name=e] {{\scriptsize $1$}};
      \draw (0,2) node[name=d] {{\scriptsize $0$}};
      \draw (-2,2) node[name=c] {{\scriptsize $-1$}};
      \draw (-4,2) node[name=b] {{\scriptsize $-2$}};
      \draw (-5,2) node[name=c1] {{$\cdots$}};
      \draw (5,2) node[name=c2] {{$\cdots$}};
      \draw (-3.5,1) node[name=c3] {{$\cdots$}};
      \draw (3.5,1) node[name=c4] {{$\cdots$}};
      \draw (-4,3) node[name=c5] {{$\cdots$}};
      \draw (4,3) node[name=c6] {{$\cdots$}};

      \draw (0,0) node[name=a] {{\scriptsize$\varnothing$}};

      \draw[semithick] (a) -- (b);
      \draw[semithick] (a) -- (c);
      \draw[semithick] (a) -- (d);
      \draw[semithick] (a) -- (e);
      \draw[semithick] (a) -- (f);
    
      \draw[semithick] (g) -- (b);
      \draw[semithick] (g) -- (c);

      \draw[semithick] (h) -- (e);
      \draw[semithick] (h) -- (f);

      \draw[semithick] (i) -- (g);
      \draw[semithick] (i) -- (d);

      \draw[semithick] (j) -- (g);
      \draw[semithick] (j) -- (h);

      \draw[semithick] (k) -- (h);
      \draw[semithick] (k) -- (d);

      \draw[semithick] (l) -- (i);
      \draw[semithick] (l) -- (j);
      \draw[semithick] (l) -- (k);
      
\end{tikzpicture}
\par
}
}

\noindent
$B$ is intended to be an abstract domain which includes both constant and sign information of an integer variable.
Indeed $B$ can be defined as reduced product of the standard constant propagation domain \cite{nnh} and of the sign abstraction
in Example~\ref{exsign}. 
For example, for a while program such as: 
$$x:=2; y:= 2; ~\textbf{while}~ x<9 ~\textbf{do}~ x:=x*y;$$
a standard analysis with this abstract domain $B$ allows us to derive the loop invariant $\{x > 0,\: y=2\}$.

It turns out that the abstraction $B$ can be constructively defined. This definition of $B$ relies on $\eta: \bZ \ra B$ and $\mu:B\ra \wp(\bZ)$ which are essentially defined 
as identity functions. It should be clear that $B$ is a purely partitioning domain, while it is not a fully 
partitioning domain, and therefore $B$ cannot be equivalently defined within the constructive Galois connection approach.
In fact, $\cC=\tuple{\eta,\bZ,B,\mu}$ is not a CGC, because
$1\in \mu(> \!0)$ while $1=\eta(1)\neq\; >\!0$.
Instead, $\cC$ turns out to be a PCGC. 

Next, consider the concrete binary integer multiplication $\otimes:\bZ\times \bZ \ra \bZ$. By following
Theorem~\ref{teoPCGCsound}~(1), we define a corresponding abstract
multiplication $\otimes_\sharp:B\times B\ra B$ as follows: 
$$\otimes_\sharp(b_1,b_2) \ud \eta^\vee(\otimes^\diamond (\mu(b_1),\mu(b_2))$$
Namely, $\otimes_\sharp$ is the best correct approximation of the powerset lifting 
$\otimes^\diamond:\wp(\bZ)\times \wp(\bZ) \ra \wp(\bZ)$ w.r.t.\ the PPGC
$\tuple{\eta^\vee,\wp(\bZ)_\subseteq, B_\leq ,\mu}=\bT_{\text{PPGC}}(\cC)$. For instance,
$\otimes_\sharp(2,<\! 0) = \:<\! 0$ and $\otimes_\sharp(-2,\leq 0)=\: \geq 0$. Then, 
since $\tuple{\otimes^\diamond,\otimes_\sharp}$ is sound, by construction, for $\bT_{\text{PPGC}}(\cC)$, we have that
$\tuple{\otimes,\otimes_\sharp}$ 
is sound for $\cC$. Furthermore, as expected, it turns out that $\otimes_\sharp$ is backward complete for $\cC$, meaning 
that for any $X,Y\in \wp(\bZ)$, $\vee_B \{x\otimes y~|~ x\in X,y\in Y\} = \otimes_\sharp( \vee_B X, \vee_B Y)$. 
For instance, we have that: 
\begin{multline*}
\vee_B (\otimes^\diamond (\{2,4\},\{-1,0\})) = \vee_B \{0,-2,-4\} = \:\; \leq\! 0\; = \\
\otimes_\sharp (>\! 0,\leq\! 0) = \otimes_\sharp(\vee_B \{2,4\},\vee_B \{-1,0\}).
\end{multline*}

\section{Conclusion}
This paper showed that constructive Galois connections, proposed by Darais and Van Horn~\cite{dv16}
as a way to define domains to be used in a mechanized and calculational style of abstract interpretation, 
are isomorphic to an already known class of Galois connections
which formalize partitions of an unordered set as an abstract domain. Building on that,  
we defined a novel  
class of constructive abstract domains for a mechanized approach to abstract interpretation, called purely constructive
Galois connections. We showed that this class of abstract domains
permits to represent a set partition with more flexibility while preserving a constructive approach to Galois connections.

\newpage
\appendix
\section{Proofs}\label{app-sec}

\rCGCprop*
\begin{proof}
Firstly, let us observe that for any $a\in A$, $a\in \mu(\eta(a))$. \\
(1) If $\mu(\eta(a_1))= \mu(\eta(a_2))$ then $a_1\in \mu(\eta(a_1))=\mu(\eta(a_2))$ so that, by (CGC-Corr), 
$\eta(a_1)=\eta(a_2)$. Next, we show that 
$\mu(\eta(a_1))=\mu(\eta(a_2)) \:\Lra\: \exists a\in A. a\in \mu(\eta(a_1))\cap \mu(\eta(a_2))$: 
on the one hand, since $a_1\in \mu(\eta(a_1))$ then $a_1\in \mu(\eta(a_2))$; on the other hand, if 
$a\in \mu(\eta(a_1))\cap \mu(\eta(a_2))$, then, by (CGC-Corr), $\eta(a_1)=\eta(a)=\eta(a_2)$, so that $\mu(\eta(a_1))=\mu(\eta(a_2))$.
\\
(2) Let us check that $\mu(b) \neq \varnothing \:\Lra\: b\in \eta(A)$: 
there exists $a\in \mu(b)$ iff, by (CGC-Corr), there exists $a\in A$ such that 
$\eta(a)=\eta(b)$ iff $b\in \eta(A)$.
\end{proof}

\rCGPprop*
\begin{proof}
(1) If $\mu(\eta(a_1))=\mu(\eta(a_2))$ then $a_1\in \mu(\eta(a_2))$ and  $a_2\in \mu(\eta(a_1))$, 
so that, by (CGP-Corr), 
$\eta(a_1)\leq_B \eta(a_2)\leq \eta(a_1)$. \\
(2) If $a\in \mu(b)$ then, by (CGP-Corr), $\eta(a)\leq_B b$, so that $\eta(a) \in \: \downarrow\!\! b \cap \eta(A)$. 
Conversely, if $b' \in \:\downarrow\!\! b \cap \eta(A)$ then $b'=\eta(a') \leq_B b$, for some $a'\in A$, so
that,  by (CGP-Corr), $a'\in \mu(b)$. 
\\
(3) Let us check that for all $X\in\wpd(A)$ and $b\in B$,
$\eta^\vee(X) \leq_B b \Lra X\subseteq \mu(b)$: $\eta^\vee(X) \leq_B b$ iff $\vee_{a\in X} \eta(a) \leq_B b$ iff
$\forall a\in X.\: \eta(a) \leq_B b$ iff $\forall a\in X.\: a \in \mu(b)$ iff $X\subseteq \mu(b)$. \\
(4) Since, by (3), $\tuple{\eta^\vee,\wpd(A),B,\mu}$ is a GC, we have that 
$\mu= \mu\circ \eta^\vee \circ \mu$, so that $\mu(B) = \mu(\etav(\wpd(A)))$ follows. 
\end{proof}

\rCGCproptwo*
\begin{proof}
$(\Ra)$ For any $a\in A$, we have that $\mu_1(\eta_1(a)) \in \mu_2(B_2)$, so that $\mu_1(\eta_1(a)) = 
\mu_2(b_2)$ for some $b_2\in B_2$. Then, since $a\in \mu_2(b_2)$, by Lemma~\ref{CGCprop}~(2), 
$b_2\in \eta_2(A)$, so that there exists some $x^2_a\in A$ such that $b_2 = \eta_2(x^2_a)$ 
and, in turn, $\mu_1(\eta_1(a)) = \mu_2(\eta_2(x^2_a))$. 
 We define 
$f_{1,2}(\eta_1(a)) \ud \eta_2(x^2_a)$. 
Dually, for any $a\in A$ there exists $x^1_a\in A$ such that $\mu_2(\eta_2(a)) = \mu_1(\eta_1(x^1_a))$,
so that we define 
$f_{2,1}(\eta_2(a)) \ud \eta_1(x^1_a)$. 
Thus, we have that for any $a\in A$, 
$\mu_1(\eta_1(a))= \mu_2(\eta_2(x^2_a)) = \mu_1(\eta_1(x^1_{x^2_a}))$, so that, 
by Lemma~\ref{CGCprop}~(1), $\eta_1(a)=\eta_1(x^1_{x^2_a})$. Dually, 
$\eta_2(a)=\eta_1(x^2_{x^1_a})$.
Hence, it turns out that 
$f_{1,2}(f_{2,1}(\eta_2(a))) = f_{1,2}(\eta_1(x^1_a)) = \eta_2(x^2_{x^1_a}) =\eta_2(a)$ and
$f_{2,1}(f_{1,2}(\eta_1(a))) = f_{2,1}(\eta_2(x^2_a)) = \eta_1(x^1_{x^2_a}) =\eta_1(a)$. 
Moreover, 
$\mu_1(\eta_1(a))= \mu_2(\eta_2(x^2_a))=\mu_2(f_{1,2}(\eta_1(a)))$ and 
$\mu_2(\eta_2(a))= \mu_1(\eta_1(x^2_a))=\mu_1(f_{2,1}(\eta_2(a)))$.
\\
$(\La)$ If $x\in \mu_1(B_1)$ then $x\in \mu_1(b_1)$, for some $b_1\in B_1$, so that, 
by Lemma~\ref{CGCprop}~(2), $b_1=\eta_1(a)$ and $x\in \mu_1(\eta_1(a))$ for some $a\in A$. Thus, 
$x\in \mu_1(\eta_1(a))=\mu_2(f_{1,2}(\eta_1(a)))$, namely, $x\in \mu_2(B_2)$. Dually, $\mu_2(B_2)\subseteq
\mu_1(B_1)$.
\end{proof}

\rteoPGC*
\begin{proof}
(1) In order to have a GC, it must be that for any $X\in \wp(A)$, $Y\in \wp(B)$, 
$\eta^\diamond(X) \subseteq Y  \Lra X\subseteq  \mu^*(Y)$. 
From left to right: if $a\in X$ then $\eta(a) \in Y$, so that, since $a\in \mu(\eta(a))$, we have that $a\in
\mu^*(Y)$. From right to left: 
if $b\in \eta^\diamond(X)$ then $b=\eta(a)$ for some $a\in X$, so that, since $a\in \mu^*(Y)$,  
we have that $a\in \mu(b')$ for some $b'\in Y$, therefore $b=\eta(a)=b'$, so that $b\in Y$. 
Let us notice that $\mu^*$ is clearly disjunctive: $\mu^*(\cup_i Y_i) = \cup_i \mu^*(Y_i)$.
Finally, $\{\mu^*(\eta^\diamond(\{a\}))~|~a\in A\} = \{\mu^*(\{\eta(a)\})~|~a\in A\} =
\{\mu(\eta(a))~|~a\in A\}$, so that, by Lemma~\ref{CGCprop}~(3-4), 
and since $\cup_{a\in A} \mu(\eta(a))=A$ it turns out that
$\{\mu(\eta(a))~|~a\in A\}$ is a partition of $A$.

\noindent
(2) First, notice that $a'\in \gamma(\alpha(\varnothing)) \Lra \alpha(\{a'\})=\alpha(\varnothing)$ holds.
Then, 
it is enough to check that $a'\in \gamma(\alpha(\{a\}) \Lra \alpha(\{a'\})=\alpha(\{a\})$, for any $a,a'\in A$. 
On the one hand, if $a'\in \gamma(\alpha(\{a\}))$ then, by PCG, $\gamma(\alpha(\{a'\})) = \gamma(\alpha(\{a\}))$,
so that $\alpha(\{a'\})=\alpha(\gamma(\alpha(\{a'\}))) = \alpha(\gamma(\alpha(\{a\})))=\alpha(\{a\})$. 
On the other hand, if $\alpha(\{a'\})=\alpha(\{a\})$ then, by CG, 
$a'\in \gamma(\alpha(\{a'\}))=\gamma(\alpha(\{a\}))$. 

\noindent
(3) If $\cC=\tuple{\eta,A,B,\mu}$ is a CGC then we have that $\cC \cong_\varnothing \bT_{\text{CGC}}(\bT_{\text{PGC}}(\cC))$ 
because $\mu^* (\{\eta^\diamond(\{a\})~|~ a\in A\}) = \{\mu^*(\{\eta(a)\})~|~a\in A\} = \{\mu(\eta(a)) ~|~a\in A\}$, 
so that, by Lemma~\ref{CGCprop}~(2), 
$\mu^* (\{\eta^\diamond(\{a\})~|~ a\in A\})\cup \{\varnothing\} =\mu(B)\cup \{\varnothing\}$. 
On the other hand, 
if $\cC=\tuple{\alpha,\wp(A)_\subseteq,D_\leq,\gamma}$ is a PGC then $\cC \cong 
\bT_{\text{PGC}}(\bT_{\text{CGC}}(\cC))$ because 
\begin{align*}
\gamma^* ( (\sing{\alpha})^\diamond (\wp(A))) & = \quad\text{[by definition of $(\sing{\alpha})^\diamond$]}\\
\gamma^* \Big( \big\{  \{ \alpha(\{a\}) ~|~ a\in X\} ~|~ X\in \wp(A)\big\}\Big) & = \quad\text{[by definition of $\gamma^*$]}\\ 
\{ \cup_{a\in X} \gamma(\alpha(\{a\}))~|~X\in \wp(A)\} & = \quad\text{[as $\gamma$ and $\alpha$ are additive]}\\ 
\{\gamma(\alpha(X))~|~X\in \wp(A)\}& =\\ 
\gamma(\alpha(\wp(A)))&
\end{align*} 
\end{proof}

\rcoroPGC*
\begin{proof}
By definition, we have that $\cC_1 \sqsubseteq \cC_2$ iff $\mu_2(B_2) \subseteq \mu_1(B_1)$,
while $\bT_{\text{PGC}}(\cC_1) \sqsubseteq \bT_{\text{PGC}}(\cC_2)$ iff $\mu_2^* (\eta_2^\diamond (\wp(A))) 
\subseteq \mu_1^* (\eta_1^\diamond (\wp(A)))$. 
As in the proof of Theorem~\ref{teoPGC}~(3), we have that $\mu_i^* (\eta_i^\diamond (\wp(A))) = \mu_i (\eta_i (A))$. 
By Lemma~\ref{CGCprop}~(2), we have that $\mu_i (B_i) = \mu_i (\eta_i (A)) \cup \{\varnothing\}$. 
Thus, $\mu_2(B_2) \subseteq \mu_1(B_1)$ iff 
$\mu_2 (\eta_2 (A)) \cup \{\varnothing\} \subseteq \mu_1 (\eta_1 (A))\cup \{\varnothing\}$ iff 
$\mu_2 (\eta_2 (A)) \subseteq \mu_1 (\eta_1 (A))$. 
\end{proof}

\rteoCGP*
\begin{proof}
(1) This is Lemma~\ref{CGP-prop}~(3).

\noindent
(2) It turns out that $\bT_{\text{CGP}}(\cG)$ is a CGP: (a) $\lambda a.\alpha(\downarrow\!\!\{a\}):A \ra D$ 
is monotone: clear,
because $a\leq a'$ implies $\downarrow\!\!\{a\} \subseteq \:\downarrow\!\!\{a'\}$, so that, by monotonicity of $\alpha$,
$\alpha(\downarrow\!\!\{a\}) \leq_D \alpha(\downarrow\!\!\{a'\})$; (b) $\gamma: D 
\ra \wpd(A)$ is monotone because $\cG$ is a GC; (c) $a\in \gamma(d) \Lra 
\alpha(\downarrow\!\!\{a\}) \leq_D d$:  in fact, $a\in \gamma(d) \Lra \:\downarrow\! a \subseteq \gamma(d) \Lra 
{\alpha(\downarrow\!\!\{a\})} \leq_D d$. 

\noindent
(3) On the one hand, 
if $\cC=\tuple{\eta,A,B,\mu}$ is a CGP then we have that $\cC = \bT_{\text{CGP}}(\bT_{\text{GC}}(\cC))$ 
because $\bT_{\text{CGP}}(\bT_{\text{GC}}(\cC)) = \tuple{\lambda a. \eta^\vee(\downarrow \!\! a), A,B,\mu}$ and
for any $a\in A$, $\eta^\vee(\downarrow \!\! a) = \vee_{x\leq a} \eta(x) = \eta(a)$, because $\eta$ is monotonic.
On the other hand, 
if $\cG=\tuple{\alpha,\wpd(A)_\subseteq,D_\leq,\gamma}$ is a GC then 
\begin{align*}
\big\{\gamma \big((\lambda a.\alpha(\downarrow\!\! a ))^\vee (Y) \big)~|~ Y\in \wpd(A)\big\} & =\quad \text{[by definition of $(\lambda a.\alpha(\downarrow\!\!\{a\}))^\vee$]}\\ 
\{ \gamma( \vee_{y\in Y} \alpha (\downarrow\!\! y)) ~|~ Y\in \wpd(A) \} & =\quad \text{[by additivity of $\alpha$]}\\ 
\{\gamma( \alpha (\cup_{y\in Y} \downarrow\!\! y )) ~|~ Y\in \wpd(A) \}  & = \quad \text{[because $Y\in \wpd(A)$]}\\
\{\gamma(\alpha(Y))~|~ Y\in \wpd(A)\} &
\end{align*}
so that $\cG \cong 
\bT_{\text{GC}}(\bT_{\text{CGP}}(\cG))$ holds.
\end{proof}

\rteoCCO*
\begin{proof}
(1) By (CGC-Corr), $x\in \mu(\eta(y))$ iff $\eta(x)=\eta(y)$ iff, by  Lemma~\ref{CGCprop}~(1), $\mu(\eta(x))=\mu(\eta(y))$.\\
(2) By definition, $x\in \varphi(a)$ iff $\varphi(x)=\varphi(a)$, and this precisely 
means that $\bT_{\text{CGC}} (\varphi)$ is a CGC. \\
(3) If $\varphi: A \ra \wp(A)$ is a CCO then $\bT_{\text{CCO}}(\bT_{\text{CGC}} (\varphi))=\varphi$. If $\cC=\tuple{\eta,A,B,\mu}$ is a CGC then 
$\cC \cong_{\varnothing} \tuple{\mu\circ \eta, A, \{\mu(\eta(a))~|~a\in A\},\id}
=\bT_{\text{CGC}}(\bT_{\text{CCO}} (\cC))$ clearly holds, by Lemma~\ref{CGCprop}~(2). 
\end{proof}

\rlemmaone*
\begin{proof}
By soundness in GCs, $g^\diamond (\gamma(\alpha(\{a\}))) \subseteq 
\gamma(g_\sharp(\alpha(\gamma(\alpha(\{a\})))))$ $= \gamma(g_\sharp(\alpha(\{a\})))$. 
Also, since $a\in \gamma(\alpha(\{a\}))$, we have that $g(a)\in g^\diamond (\gamma(\alpha(\{a\})))$. 
Since $g_\sharp$ is block preserving,
$g(a)\in g^\diamond (\gamma(\alpha(\{a\})))\subseteq \gamma(g_\sharp(\alpha(\{a\}))) = \gamma(\alpha(\{a'\}))$, for some
$a'\in A$. Hence, since $g(a)\in \gamma(\alpha(\{a'\}))$ and $\gamma(\alpha(\{a'\}))$ is a block, we have that 
$\gamma(\alpha(\{g(a)\})) = \gamma(\alpha(\{a'\}))$, and, in turn,  $\gamma(\alpha(\{g(a)\}))= \gamma(g_\sharp(\alpha(\{a\})))$,
so that  $\alpha(\gamma(\alpha(\{g(a)\})))= \alpha(\gamma(g_\sharp(\alpha(\{a\}))))$. 
By GI, $\alpha\circ \gamma = \id$, consequently  $g_\sharp(\alpha(\{a\})) = \alpha(\{g(a)\})$. Also, 
$g^\diamond (\gamma(\alpha(\{a\}))) \subseteq \gamma(g_\sharp(\alpha(\{a\}))) = \gamma(\alpha(\{g(a)\}))$.
\end{proof}

\rtheosound*
\begin{proof}
(1) We have that 
$\tuple{f,f_\sharp}$ is sound iff $\forall a\in A.\: f(a) \in \mu(f_\sharp(\eta(a)))$. 
Also, we have that $\bT_{\text{PGC}}(\tuple{f,f_\sharp})$ is sound w.r.t.\ $\bT_{\text{PGC}}(\cC)$ iff 
$\forall X\in \wp(A).\: f^\diamond (X) \subseteq \mu^* (f_\sharp^\diamond (\eta^\diamond (X)))$ 
iff $\{f(x)~|~x\in X\} \subseteq \cup_{y\in X} \mu(f_\sharp (\eta(y)))$ iff
$\forall X\in \wp(A).\forall x\in X.\exists y\in X. f(x)\in \mu(f_\sharp(\eta(y)))$. 
Then, $\forall a\in A.\: f(a) \in \mu(f_\sharp(\eta(a)))$ iff
$\forall X\in \wp(A).\forall x\in X.\exists y\in X. f(x)\in \mu(f_\sharp(\eta(y)))$ clearly holds.  \\
(2) $\tuple{g^\diamond,g_\sharp}$ is sound iff $\forall X\in \wp(A).\: g^\diamond(X) \subseteq  \gamma(g_\sharp(\alpha(X)))$
iff $\forall X\in \wp(A).\:\{g(x)~|~x\in X\} \subseteq \gamma(g_\sharp(\alpha(X)))$. 
Moreover, $\bT_{\text{CGC}}(\tuple{g^\diamond,g_\sharp})=\tuple{g,g_\sharp^r}$ is sound 
for $\tuple{\sing{\alpha},A,\{\alpha(\{a\})~|~a\in A\},\gamma}$
iff
$\forall a\in A.\: g(a) \in \gamma(g_\sharp^r (\sing{\alpha} (a))) = \gamma(g_\sharp^r (\alpha(\{a\}))) =  
\gamma(g_\sharp (\alpha(\{a\})))$, where the latter equality follows by Lemma~\ref{lemma-one}.
Then, if $\tuple{g^\diamond,g_\sharp}_\cD$ is sound then 
$\bT_{\text{CGC}}(\tuple{g^\diamond,g_\sharp})$ is clearly sound. On the other hand, if $\bT_{\text{CGC}}(\tuple{g^\diamond,g_\sharp})$
is sound and $X\in \wp(A)$ then $\{g(x)~|~x\in X\} \subseteq \cup_{x\in X} \gamma(g_\sharp (\alpha(\{x\})))
\subseteq \gamma(g_\sharp (\alpha(X)))$, where the latter containment follows 
by monotonicity of $\gamma$, $g_\sharp$ and $\alpha$. 
\\
(3) Let $\tuple{f,f^\sharp}$ be a sound pair for a CGC $\cC = \tuple{\eta,A,B,\mu}$. By (1) we know that 
$\tuple{f^\diamond,f_\sharp^\diamond} = \bT_{\text{PGC}}(\tuple{f,f_\sharp})$ is sound for
$\bT_{\text{PGC}}(\cC)=\tuple{\eta^\diamond,\wp(A)_\subseteq,\wp(B)_\subseteq,\mu^*}$.
Moreover, $f_\sharp^\diamond: \wp(B)\ra \wp(B)$ is monotone and block-preserving, meaning that 
$f_\sharp^\diamond (\eta^\diamond(\{a\})) %
=\eta^\diamond(\{a'\})$ for some $a'\in A$, 
because, by soundness, $f(a) \in \mu(f_\sharp(\eta(a)))$, so that $\eta(f(a)) = f_\sharp(\eta(a))$,
thus implying that $f_\sharp^\diamond (\eta^\diamond(\{a\})) = \{f_\sharp(\eta(a))\}=\{\eta(f(a))\}=\eta^\diamond(\{f(a)\})$. 
Let us consider $\tuple{f,(f_\sharp^\diamond)^r} = \bT_{\text{CGC}}(\bT_{\text{PGC}}(\tuple{f,f_\sharp}))$,
where $\bT_{\text{CGC}}(\bT_{\text{PGC}}(\cC)) = \tuple{\sing{(\eta^{\diamond})},A, \{\eta^\diamond(\{a\})~|~a\in A\},\mu^*}$.
We have that
$(f_\sharp^\diamond)^r: \{\eta^\diamond(\{a\})~|~a\in A\} \ra \{\eta^\diamond(\{a\})~|~a\in A\}$ is such that 
$(f_\sharp^\diamond)^r (\eta^\diamond(\{a\})) = 
\eta^\diamond (f (a))=f_\sharp^\diamond (\eta^\diamond(\{a\}))$, so we have that 
$\mu^* ((f_\sharp^\diamond)^r (\sing{(\eta^{\diamond})} (a))) = 
\mu^* ((f_\sharp^\diamond)^r (\eta^{\diamond}(\{a\}))) =
\mu^* ( f_\sharp^\diamond (\eta^\diamond(\{a\})))=
\mu(f_\sharp(\eta(a)))$. 
Hence, this shows that $\bT_{\text{CGC}}(\bT_{\text{PGC}}(\tuple{f,f_\sharp}))\cong \tuple{f,f_\sharp}$.

\noindent
On the other hand, let $\tuple{g^\diamond,g^\sharp}$ be a sound pair for a PGC $\cD = \tuple{\alpha,\wp(A)_\subseteq,D_\leq,\gamma}$, 
where $g:A\ra A$ and $g_\sharp$ is monotone and block-preserving. 
By (2),  
$\tuple{g,g_\sharp^r} = \bT_{\text{CGC}}(\tuple{g^\diamond,g_\sharp})$ is sound for the CGC 
$\bT_{\text{CGC}}(\cD)=\tuple{\sing{\alpha},A,\{\alpha(\{a\})~|~a\in A\},\gamma}$.
We therefore consider $\tuple{g^\diamond,(g_\sharp^r)^\diamond} = \bT_{\text{PGC}}(\bT_{\text{CGC}}(\tuple{g,g_\sharp}))$ which,
by (1), is sound for $\bT_{\text{PGC}}(\bT_{\text{CGC}}(\cD)) =
\tuple{(\sing{\alpha})^\diamond,\wp(A)_\subseteq, \wp(\{\alpha(\{a\})~|~a\in A\}),\gamma^*}$. 
We have to check that $\gamma \circ g_\sharp \circ \alpha = \gamma^* \circ  (g_\sharp^r)^\diamond \circ (\sing{\alpha})^\diamond$. For any $X\in \wp(A)$, we have that 
$(g_\sharp^r)^\diamond  ((\sing{\alpha})^\diamond (X)) = 
(g_\sharp^r)^\diamond (\{\alpha(\{a\})~|~a\in X\}) = \{g_\sharp(\alpha(\{a\}))~|~a\in X\}$, so that
$\gamma^* ((g_\sharp^r)^\diamond  ((\sing{\alpha})^\diamond (X))) = \cup_{a\in X} \gamma(g_\sharp(\alpha(\{a\})))$. 
Since $g_\sharp$ and $\gamma$ are additive, while $\alpha$ is always additive, 
we have that $\cup_{a\in X} \gamma(g_\sharp(\alpha(\{a\}))) = \gamma(g_\sharp(\alpha(X)))$. 
\end{proof}

\rlemsound*
\begin{proof}
Let us furst recall that $\bT_{\text{PGC}}(\cC)= 
\tuple{\eta^\diamond,\wp(A)_\subseteq,\wp(B)_\subseteq,\mu^*}$.\\
(1)~$\tuple{f^\diamond,f_\sharp^\diamond}$ is BCA w.r.t.\  $\bT_{\text{PGC}}(\cC)$ iff $f_\sharp^\diamond = \eta^\diamond \circ f^\diamond \circ 
\mu^*$ iff $\forall Y\in \wp(B). \{f_\sharp(y)~|~y\in Y\} = \eta^\diamond(f^\diamond(\cup_{y\in Y} \mu(y))) = 
\{\eta(f(x))~|~x\in \mu(y), y\in Y\}$  iff $\tuple{f,f_\sharp}$ satisfies {\rm (CGC-Cmp/$\eta\mu$)}. 
\\
(2)~$\tuple{f^\diamond,f_\sharp^\diamond}$ is forward complete w.r.t.\  $\bT_{\text{PGC}}(\cC)$ 
iff $f^\diamond \circ \mu^* = \mu^* \circ f_\sharp^\diamond$ iff $\forall Y\in \wp(B). \{f (x)~|~x\in \mu(y), y\in Y\} = 
\cup \{\mu(f_\sharp(y))~|~y\in Y\}$ iff $\tuple{f,f_\sharp}$ satisfies {\rm (CGC-Cmp/$\mu\mu$)}. \\
(3)~$\tuple{f^\diamond,f_\sharp^\diamond}$ is backward complete w.r.t.\  $\bT_{\text{PGC}}(\cC)$ 
iff $\eta^\diamond\circ f^\diamond  = f_\sharp^\diamond\circ \eta^\diamond$ iff $\forall X\in \wp(A). 
\{\eta(f (x))~|~x\in X\} = \{f_\sharp(\eta(x))~|~x\in X\}$ iff $\tuple{f,f_\sharp}$ satisfies {\rm (CGC-Cmp/$\eta\eta$)}. \\
(4)~$\tuple{f^\diamond,f_\sharp^\diamond}$ is precise w.r.t.\ $\bT_{\text{PGC}}(\cC)$ 
iff $f^\diamond  = \mu^* \circ f_\sharp^\diamond\circ \eta^\diamond$ iff $\forall X\in \wp(A). 
\{f (x)~|~x\in X\} = \cup \{\mu(f_\sharp(\eta(x)))~|~x\in X\}$ iff $\tuple{f,f_\sharp}$ satisfies {\rm (CGC-Cmp/$\mu\eta$)}. 
\end{proof}

\rPCGCprop*
\begin{proof}
(1) $a\in \mu(\wedge_i b_i)$ iff, by PCGC~(2), $\eta(a) \leq \wedge_i b_i$ iff
$\forall i.\, \eta(a) \leq b_i$ iff, by PCGC~(2), $\forall i.\, a\in \mu(b_i)$ iff $a\in \cap_i \mu(b_i)$.
\\
(2) The proof of $\eta(a_1)= \eta(a_2) \:\Lra\: \mu(\eta(a_1))=\mu(\eta(a_2))$, by PCGC~(1), is the same of Lemma~\ref{CGP-prop}~(1). 
Also, $\mu(\eta(a_1))=\mu(\eta(a_2)) \: \Ra\: a_1,a_2\in \mu(\eta(a_1))\cap \mu(\eta(a_2))$. Conversely,
if $a\in \mu(\eta(a_1))\cap \mu(\eta(a_2))$ then, by PCGC~(1), $\eta(a_1)=\eta(a)=\eta(a_2)$. \\
(3) If $b\in \eta(A)$ then $b=\eta(a)$, for some $a$, so that $a\in \mu(\eta(a))$. A counterexample to the vice versa
is in Example~\ref{ex-PCGC}.
\\
(4) $\vee_{a\in X} \eta(a)\leq b \:\Lra\: \forall a\in X.\eta(a)\leq b \:\Lra\: \forall a\in X.
a\in \mu(b) \:\Lra\: X\subseteq \mu(b)$.
\end{proof}

\rteoPCGC*
\begin{proof}
(1) By Lemma~\ref{PCGC-prop}~(3), $\bT_{\text{PPGC}}(\cC) = \tuple{\eta^\vee,\wp(A)_\subseteq, B_\leq,\mu}$ is a GC. If
$a\in \mu(\eta^\vee(\{a_1\}) \cap \mu(\eta^\vee(\{a_2\}))$ then, by PCGC~(1), $\eta(a_1)=\eta(a)=\eta(a_2)$, 
so that $\mu(\eta^\vee(\{a_1\})) = \mu(\eta^\vee(\{a_2\}))$, implying that $\bT_{\text{PPGC}}(\cC)$ is indeed a PPGC.  
\\
(2)  Let us show PCGC~(1): 
\begin{align*}
a\in \gamma(\sing{\alpha}(a')) & \Ra \quad\text{[by definition]}\\
a\in \gamma(\alpha(\{a'\})) & \Ra \quad\text{[by GC]}\\
\alpha(\{a\}) \leq \alpha(\{a'\}) & \Ra \quad\text{[by monotonicity of $\gamma$]}\\
\gamma(\alpha(\{a\}))\subseteq \gamma(\alpha(\{a'\})) & \Ra \quad\text{[by PPGC]}\\
\gamma(\alpha(\{a\}))= \gamma(\alpha(\{a'\})) & \Ra \quad\text{[by GC]}\\
\alpha(\{a\})= \alpha(\{a'\}) & \Ra \quad\text{[by definition]}\\
\alphasing(a)= \alphasing(a') &
\end{align*}
On the other hand, $\alphasing(a)= \alphasing(a')$ means $\alpha(\{a\})= \alpha(\{a'\})$, so that, by GC,
$a\in \gamma(\alpha(\{a\}))= \gamma(\alpha(\{a'\})) = \gamma(\sing{\alpha}(a'))$. 
PCGC~(2) also follows because $a\in \gamma(d)$ iff $\{a\} \subseteq \gamma(d)$ iff, by GC,  
$\alpha(\{a\}) \leq d$ iff $\alphasing(a) \leq d$. 
\\
(3) Let $\cC = \tuple{\eta,A,B_\leq,\mu}$ be a PCGC. We have that  
$\bT_{\text{PCGC}}(\bT_{\text{PPGC}}(\cC)) = \tuple{\sing{(\eta^\vee)},A,B_\leq,\mu}$, so that 
$\bT_{\text{PCGC}}(\bT_{\text{PPGC}}(\cC)) = \cC$ holds because $\sing{(\eta^\vee)}=\eta$.
On the other hand, let $\cD=\tuple{\alpha,\wp(A)_\subseteq,D_\leq,\gamma}$ be a PPGC. Here, 
$\bT_{\text{PPGC}}(\bT_{\text{PCGC}}(\cD)) = \tuple{(\alphasing)^\vee,\wp(A)_\subseteq,D_\leq,\gamma}$. 
We have that 
\begin{align*}
\gamma((\alphasing)^\vee (\wp(A))) &=\\ 
\{\gamma((\alphasing)^\vee (X))~|~X\in \wp(A)\} &=\quad\text{[by definition]}\\ 
\{\gamma(\vee_{a\in X} \alphasing(a)) ~|~X\in \wp(A)\} & =\quad\text{[by definition]}\\
\{\gamma(\vee_{a\in X} \alpha(\{a\})) ~|~X\in \wp(A)\} & =\quad\text{[since $\alpha$ is additive]}\\
\{\gamma(\alpha(X)) ~|~X\in \wp(A)\} & =\\
\gamma(\alpha(\wp(A)) &
 \end{align*}
so that $\bT_{\text{PPGC}}(\bT_{\text{PCGC}}(\cD)) \cong \cD$. 
\end{proof}

\rteoPCGCsound*
\begin{proof}
(1) We have that $\tuple{f,f_\sharp}$ is sound w.r.t.\ $\cC$ iff $\eta^\vee \circ  f^\diamond \circ 
\mu \sqsubseteq f_\sharp$ iff $\tuple{f^\diamond,f_\sharp}$ is sound w.r.t.\ $\bT_{\text{PPGC}}(\cC)=
\tuple{\eta^\vee,\wp(A)_\subseteq, B_\leq,\mu}$.\\
(2) $\tuple{g^\diamond,g_\sharp}$ is sound w.r.t.\ $\cD$ iff $\alpha \circ g^\diamond \circ  \gamma 
\sqsubseteq g_\sharp$, 
while  
$\tuple{g,g_\sharp}=\bT_{\text{PCGC}}(\tuple{g^\diamond,g_\sharp})$ is sound
w.r.t.\ $\bT_{\text{PCGC}}(\cD) = \tuple{\alphasing,A,D_\leq,\gamma}$ iff
$g^\diamond_{\bT_{\text{PCGC}}(\cD)} \sqsubseteq g_\sharp$ iff 
$(\alphasing)^\vee \circ g^\diamond \circ \gamma \sqsubseteq g_\sharp$. 
Then, it is enough to observe that $(\alphasing)^\vee = \alpha$: in fact, 
for any $X\in\wp(A)$, 
$(\alphasing)^\vee (X) = \vee_{x\in X} \alphasing(x) = \vee_{x\in X} \alpha(\{x\}) = \alpha(\cup_{x\in X} \{x\}) = \alpha(X)$. 
\\
(3) Clear.  
\end{proof}


\begin{thebibliography}{99}

\bibitem{blazy13}
S.~Blazy, V.~Laporte, A.~Maroneze, D.~Pichardie. 
\newblock Formal verification of a {C} value analysis based on abstract interpretation. 
\newblock In \emph{Proc. of the International Static Analysis Symposium (SAS'13)}, Springer 
LNCS 7935, pp.~324-344, 2013.

\bibitem{cou99}
P.~Cousot. 
\newblock The calculational design of a generic abstract interpreter. 
\newblock In M.~Broy and R.~Steinbr\"{u}ggen, eds, \emph{Calculational System Design}, 
NATO ASI Series F. IOS Press, Amsterdam, 1999.


\bibitem{CC77}
P.~Cousot and R.~Cousot.
\newblock Abstract interpretation: a unified lattice model for static analysis
  of programs by construction or approximation of fixpoints.
\newblock In \emph{Proc.\ 4th ACM POPL}, pp.~238-252, 1977.

\bibitem{CC79}
P.~Cousot and R.~Cousot.
\newblock Systematic design of program analysis frameworks.
\newblock In \emph{Proc.\ 6th ACM POPL}, pp.~269-282, 1979.

\bibitem{cc92}
P.~Cousot and R.~Cousot. 
\newblock Abstract interpretation frameworks.  
\newblock \emph{J.\ Logic and Computation},  2(4):511-547, 1992.  


\bibitem{CC94}
P.~Cousot and R.~Cousot.
\newblock Higher-order abstract interpretation (and application to comportment
  analysis generalizing strictness, termination, projection and {P}{E}{R}
  analysis of functional languages) ({I}nvited {P}aper).
\newblock In {\em Proc.\ of the IEEE Int.\ Conf.\ on Computer Languages 
(ICCL'94)}, pp.~95-112. IEEE Computer Society Press, 1994.

\bibitem{ch78}
P.~Cousot and N.~Halbwachs.
\newblock Automatic discovery of linear restraints among variables
of a program. 
\newblock In \emph{Proc.\ 5th ACM POPL}, pp.~84-97, 1978.


\bibitem{dv16}
D.~Darais and D.~Van Horn. 
\newblock Constructive Galois connections: taming the Galois connection framework for mechanized metatheory. 
\newblock In \emph{Proceedings of the 21st ACM Intern.\ Conf.\ on Functional Programming (ICFP'16)}. 
ACM, pp.~311-324, 2016.  


\bibitem{gr98}
R.~Giacobazzi and F.~Ranzato.
\newblock Optimal domains for disjunctive abstract interpretation. 
\newblock \emph{Sci.\ Comp.\ Program.}, 32:177-210, 1998.

\bibitem{gr14}
R. Giacobazzi and F. Ranzato. 
\newblock Correctness kernels of abstract interpretations. 
\newblock \emph{Information and Computation}, 237:187-203, 2014.

\bibitem{grs00}
R.~Giacobazzi, F.~Ranzato and F.~Scozzari.
\newblock Making abstract interpretations complete.
\newblock \emph{J.~ACM}, 47(2):361-416, 2000.

\bibitem{jou}
J.H.~Jourdan. 
\newblock \emph{Verasco: a Formally Verified {C} Static Analyzer}. 
\newblock PhD thesis, Universit\'e Paris Diderot (Paris 7), France, 2016.

\bibitem{verasco}
J.H.~Jourdan, V.~Laporte, S.~Blazy, X.~Leroy, D.~Pichardie. 
\newblock A formally-verified {C} static analyzer. 
\newblock In \emph{Proc.\ 42nd ACM POPL}, 
pp.~247-259, 2015.

\bibitem{mj08}
J.~Midtgaard and T.~Jensen. 
\newblock A calculational approach to control-flow analysis by abstract interpretation. 
\newblock In \emph{Proc.\ Intern.\ Static Analysis Symposium (SAS'08)}, LNCS 5079, pp.~347-362, 2008.

\bibitem{mon98}
D.~Monniaux. 
\newblock \emph{R\'{e}alisation m\'{e}canis\'{e}e d'interpr\'{e}teurs abstraits}. 
\newblock Rapport de DEA, Universit\'{e} Paris VII, France, 1998. In French.

\bibitem{nnh}
F.~Nielson, H.R.~Nielson, C.~Hankin. 
\newblock {\em Principles of Program Analysis}.
\newblock Springer-Verlag, 1999. 

\bibitem{pic05}
D.~Pichardie. 
\newblock \emph{Interpr\'{e}tation abstraite en logique intuitionniste: extraction d'analyseurs Java certifi\'{e}s}. 
PhD thesis, Universit\'{e} de Rennes, France, 2005. In French. 

\bibitem{rt05}
F.~Ranzato and F.~Tapparo.
\newblock An abstract interpretation-based refinement algorithm for
strong preservation. 
\newblock In \emph{Proc.\ 11th Intern.\ Conf.\ on Tools and Algorithms
for the Construction and Analysis of Systems (TACAS'05)}, LNCS~3440, pp.~140--156,
Springer, 2005.


\bibitem{rt07}
F.~Ranzato and F.~Tapparo.
\newblock Generalized strong preservation by abstract interpretation.
\newblock \emph{J.\ Logic and Computation}, 17(1):157-197, 2007. 

\bibitem{so08}
P.F.~Silva and J.N.~Oliveira. 
\newblock `Galculator': Functional prototype of a Galois-connection based proof assistant. 
\newblock In \emph{Proc.\ 10th International ACM Conference on Principles and Practice of Declarative Programming  (PPDP'08)}. 
ACM, 2008.

\end{thebibliography}
\end{document}